\documentclass[11pt,preprint,authoryear]{elsarticle}
\usepackage[final]{microtype}

\usepackage[T1]{fontenc}
\usepackage[utf8]{inputenc}

\usepackage{amsthm}

\usepackage{bm,amsmath}
\usepackage{amsfonts}
\usepackage{amssymb}
\usepackage{mathrsfs}
\usepackage{upgreek}
\usepackage{makecell}

\usepackage{bbm}

\usepackage{float}
\usepackage{array}
\usepackage{booktabs}
\usepackage{threeparttable}
\usepackage{longtable}
\usepackage{tabularx}
\usepackage{graphicx}
\usepackage{multicol}
\usepackage{multirow}
\usepackage{rotating}
\usepackage{lscape}
\usepackage{epstopdf}
\usepackage[center]{caption}
\usepackage[table]{xcolor}
\usepackage{subfig}
\usepackage{chngcntr}
\counterwithout{figure}{section}
\counterwithout{table}{section}

\newcolumntype{P}[1]{>{\centering\arraybackslash}p{#1}}
\newcolumntype{x}[1]{%
	>{\centering\hspace{0pt}}p{#1}}%

\usepackage[round]{natbib}

\usepackage{paralist}
\usepackage{color}
\usepackage[title]{appendix}
\usepackage{soul} 

\usepackage[left=0.9in,right=0.9in,top=1in,bottom=1in]{geometry}
\usepackage{setspace}
\usepackage{footmisc}

\usepackage{lipsum}

\RequirePackage[OT1]{fontenc}

\definecolor{sangre}{rgb}{0.6,0.18,0.19}
\definecolor{dullmagenta}{rgb}{0.6,0,0.6}
\definecolor{darkblue}{rgb}{0,0,0.7}
\definecolor{verdeprofundo}{rgb}{0.02,0.376,0.031}
\definecolor{olivegreen}{RGB}{13, 165, 100}
\definecolor{chocolate}{RGB}{33,33,33}
\definecolor{lightcyan}{rgb}{0.6,1,1}
\definecolor{purple}{RGB}{147, 81, 209}

\ifx\pdfoutput\undefined
\usepackage[hypertex,colorlinks,urlcolor = darkblue,citecolor = gary,linkcolor=gray]{hyperref}
\else
\usepackage[pdftex,hypertexnames=false,colorlinks,urlcolor = darkblue,citecolor = gray,linkcolor=gray]{hyperref}
\fi

\def\0{\phantom{0}}

\newcommand{\eps}{\varepsilon}

\newcommand{\mbb}{\mathbb}
\newcommand{\mcl}{\mathcal}

\renewcommand{\P}{\mbb{P}}

\renewcommand{\H}{\mcl{H}}

\newcommand{\1}{\mathbbm{1}}

\DeclareMathAccent{\verywidehat}{\mathord}{largesymbols}{'144}

\DeclareMathOperator{\diag}{diag}

\newtheorem{remark}{Remark}
\newtheorem{example}{Example}

\newtheorem{prop}{Proposition}[section]
\newtheorem{lemma}{Lemma}

\newtheorem{assump}{Assumption}

\newcounter{assumpl}
\makeatletter
\newcommand*{\rom}[1]{\expandafter\@slowromancap\romannumeral #1@}

\newcolumntype{L}[1]{>{\raggedright\arraybackslash}p{#1}}
\newcolumntype{C}[1]{>{\centering\arraybackslash}p{#1}}
\newcolumntype{R}[1]{>{\raggedleft\arraybackslash}p{#1}}

\newcommand\T{\rule{0pt}{3.0ex}}
\newcommand\B{\rule[-1.2ex]{0pt}{0pt}}

\newcommand{\beq}{\begin{equation}}
	\newcommand{\eeq}{\end{equation}}
\newcommand{\bol}{\begin{enumerate}}
	\newcommand{\eol}{\end{enumerate}}
\newcommand{\bul}{\begin{itemize}}
	\newcommand{\eul}{\end{itemize}}
\newcommand{\bdl}{\begin{description}}
	\newcommand{\edl}{\end{description}}
\newcommand{\bil}{\begin{inparaenum}[(i)]}
	\newcommand{\eil}{\end{inparaenum}}
\newcommand{\li}{\item}

\def\fns{\footnotesize}
\def\0{\fns{$^{\ast\ast}$}}
\def\5{\fns{$^\ast$}}
\def\10{\fns$^{\dag}$}

\begin{document}
\renewcommand*{\thefootnote}{\fnsymbol{footnote}}

\title{Breaking news}

\author[1]{Lars Winkelmann}
\author[2]{Wenying Yao}
\address[1]{School of Business and Economics, Freie Universit\"{a}t Berlin}
\address[2]{Melbourne Business School, University of Melbourne}

\begin{frontmatter}
\date{\today}

\begin{abstract}
\vspace{-0.75cm}

This paper examines how regulatory interventions in high-frequency financial markets affect price discovery.
We focus on {\em Breaking news}, where dynamic circuit breakers trigger trading halts immediately after the release of macroeconomic fundamentals. 
Within a high-frequency signal-in-noise model, we show that triggering rules complicate statistical inference for the price impact of news, rendering conventional non-parametric jump estimators inconsistent. 
Building on this insight, we develop a regression-based test for fundamental pricing that accounts for non-vanishing transition times. The test compares transition price changes to efficient jumps implied by observable factors. Our empirical analysis of CME E-mini S\&P 500 futures shows that Breaking news are associated with systematic deviations from fundamental pricing, predominantly in the form of overshooting. 
Our findings highlight a regulatory trade-off: the appeal of simple and transparent circuit breaker rules must be weighed against their cost of preventing fundamentals from being priced contemporaneously, thereby creating adverse incentives and introducing distortions.

\vspace{0.5cm}

 \vspace{0.5cm}

\end{abstract}

\begin{keyword}
Market microstructure \sep Fundamental pricing \sep Circuit breaker \sep Jumps \sep Drift \sep Persistent noise  \vspace{0.5cm}
\end{keyword}
\end{frontmatter}

\thispagestyle{plain}

\newpage 
\section{Introduction}

Financial markets aggregate information into prices. In frictionless settings, when new information about economic fundamentals arrives, prices adjust rapidly, often instantaneously, to a new level that reflects the market’s collective assessment of the news. This benchmark underlies much of modern asset pricing and event-study analysis.

In practice, however, prices at modern electronic exchanges do not move freely. Trading is governed by a complex system of price limits and circuit breakers that mechanically constrain how fast and how far prices are allowed to move.\footnote{Backed by Securities and Exchange Commission (SEC) Rule 80B and Commodity Futures Trading Commission (CFTC) regulation Part 38, for example.} As an example, Table \ref{rules} illustrates trading rules for the world's largest equity index futures contract, the E-mini S\&P 500, at the Chicago Mercantile Exchange (CME).\footnote{The E-mini S\&P 500 futures are the largest equity index futures in the world by trading value (notional value in 2024 is 111T USD): \href{https://www.fia.org/fia/articles/2024-annual-etd-volume-review}{FIA (2025)}.} The first three rules are commonly referred to as dynamic circuit breakers. They constrain price movements for intraday look-back windows ranging from one millisecond to one hour and trigger temporary trading halts whenever prices violate a limit that is defined relative to a reference price.

These rules were originally introduced to prevent severe market disruptions, such as the 1987 Black Monday or the 2010 Flash Crash. Today, the tighter bands also serve as guardrails for continuous pricing to control execution risks. A common argument made by regulators and risk managers is that circuit breakers help to maintain orderly markets and to facilitate price discovery and trading at rational prices.\footnote{\href{https://archive.org/details/reportofpresiden01unit/page/66/mode/2up}{Brady Report (1988, p66)}, and \href{https://www.sec.gov/news/press/2010/2010-98.htm\#:~:text=FOR\%20IMMEDIATE\%20RELEASE,in\%20a\%20five\%2Dminute\%20period.}{SEC release 2010-98}.} Yet the rules themselves do not distinguish between fundamentally justified and unjustified price movements.  Any move that breaches the limits is mechanically curtailed, regardless of its informational content. When large fundamental shocks occur, the same limits apply. Consequently, to avoid trading rule violations, the new information must be priced in small steps, generating a smooth transition to a new price level. 
This raises a central question for market design and regulation: does the triggering of dynamic circuit breakers during price transition only delay price discovery, or does it also distort the pricing of fundamentals?

\begin{table}[ht]
\centering
\begin{threeparttable}
\caption{Trading rules in CME E-mini S\&P 500 futures (July 2024) \label{rules}}
\begin{tabular}{l@{\hskip 0.25cm}L{6.3cm}@{\hskip 0.5cm}l@{\hskip 0.3cm}l@{\hskip 0.3cm}l} 
\toprule \hline
    & Type                                             & Look-back window           & Limit &Break time\T\B \\ \hline
(1) & Velocity logic (narrow)                          & 1 millisecond              & $\pm$6 ticks  &5 seconds\T \\
(2) & Velocity logic (wide)                            & 1 second                   & $\pm$18 ticks  & 5 seconds  \\
(3) & Dynamic circuit breaker (Non-U.S. trading hours) & 1 hour                     & $\pm$3.5\%  & 2 minutes \\
(4) & Price limits (Non-U.S. trading hours)            & 3 p.m. CT fixing price     & $\pm$7\%    & N/A \\
(5) & Market-wide circuit breaker (U.S. trading hours) & previous day's closing price & $-$7\%    & 10 minutes \\
(6) & Market-wide circuit breaker (U.S. trading hours) & previous day's closing price & $-$13\%   & 10 minutes\\
(7) & Market-wide circuit breaker (U.S. trading hours) & previous day's closing price & $-$20\%   & Rest of day\B \\ \hline 
\bottomrule
\end{tabular}
\begin{tablenotes}[flushleft]
\footnotesize
\item Notes: The tick size is 0.25 of an index point valued 50 USD (1 tick is 12.50 USD). Limits of Rules (1)--(2) are updated monthly. Rules (3)--(4) are only active during 5:00 p.m. to 8:30 a.m. Central Time (non-U.S. trading hours). Market-wide circuit breaker specified in (5)--(7) lead to a trading halt in the S\&P 500 futures only when a regulatory halt occurs in the cash equity market due to downsize move in the underlying S\&P 500 index. 
Source: \href{https://www.cmegroup.com/trading/equity-index/sp-500-price-limits-faq.html?redirect=/trading/equity-index/faq-sp-500-price-limits.html}{CME} website. 
\end{tablenotes}
\end{threeparttable}
\end{table}

At a first glance, the potential distortions may appear small. According to Table \ref{rules}, the break in matching orders induced by a velocity-logic event is often only a few seconds. However, unintended side effects of market interventions through circuit breakers have been discussed in the literature. These side effects include self-reinforcing price dynamics, such as the magnet or lock-out effects, which may change trading and quoting behavior. \citet{s94} shows that the prospect of circuit breaker can prompt traders to shift volume earlier in the day, commonly referred to as the magnet effects. \citet{cpwx24} document that volatility rises with the probability of a trading break. \citet{bla24} analyze coordination failures and market runs, proposing the incorporation of forward-looking component into circuit breakers. \citet{hh19} find that triggering events elevate volatility and produce undesirable shifts in prices. Most existing empirical studies rely on difference-in-differences-type methodologies that implicitly assume fundamentals either remain constant or cancel out. 
Since efficient prices are not directly observed, such assumptions are difficult to verify.

This paper proposes a new empirical strategy to assess the effect of dynamic circuit breakers on fundamental pricing. We focus on scheduled macroeconomic news announcements that generate measurable, discrete shifts in efficient prices. A large literature shows that macroeconomic surprises and investor attention are key factors in pricing such news \citep{abdv03, bfv21, wz22, fms22}. Building on this evidence, we model jumps in efficient prices at news releases as functions of observable news factors and develop an event-wise test of whether changes in the transition price are consistent with the efficient benchmarks. Our main interest is in \textit{Breaking-news} events, defined as announcements for which a dynamic circuit breaker interrupts trading shortly after the release time. Testing fundamental pricing across Breaking and non-Breaking news allows us to isolate the impact of circuit-breaker triggering on price discovery.

Besides its empirical motivation, the paper makes the following methodological contributions. We embed exchange trading rules such as those in Table \ref{rules} into a continuous-time signal-in-noise model by allowing the observed price to adjust over a fixed transition window following the news arrival. The transition noise can account for situations where a jump in the efficient price may go beyond the trading limits and hence cannot be priced instantaneously. Our framework requires little regularity for the transition function and covers cases including the persistent noise model of \citet{altz23}.    
In this setting, we show that conventional nonparametric jump estimators, like the pre-average estimator of \cite{lm12} and the spectral estimator of \cite{bnw19}, are no longer consistent. The inconsistency stems from the fluctuations of the efficient price accumulated over the fixed transition interval. The test for fundamental pricing builds on an inconsistent pre-average event return estimator and its distance to the jump in the efficient price. Critical values are derived from concentration inequalities. We jointly bound fluctuations of the efficient price over the transition window and typical pre-averaging estimation errors. A feasible version of the test is obtained by estimating the efficient jump from cross-event regressions.

Locally dominating drifts in high-frequency data and the drift-burst hypothesis are discussed by \cite{cor22} and \cite{sp25}, for example. While drift-burst aims at explaining a similar local phenomenon in observed high-frequency data as the persistent noise model, the former assigns the drift to the efficient price process. Economic reasoning for local drifts often builds on a learning process, during which markets discover the appropriate price level gradually. In contrast, the present paper studies a different mechanism of the price transition without learning involved. That is, at the time when macroeconomic news is released, markets know immediately and exactly the new level of the efficient price. However, trading rules, such as those in Table \ref{rules}, prevent prices to fully adjust instantaneously. The aim is to investigate if trading rule violations during price transitions have an impact on the pricing of fundamentals.  

We assess the finite-sample properties of the proposed procedures in a simulation study. The simulation results confirm a few theoretical insights. First, we demonstrate that conventional pre-average jump estimators suffer from sizable bias, whereas the regression-based estimator substantially reduces this bias. In addition, the feasible test for fundamental pricing maintains close-to-nominal size and exhibits meaningful power under the alternative, provided that transition times are not too long.

The testing framework is then applied to high-frequency CME E-mini S\&P 500 futures data around U.S. inflation announcements between 2020 and 2025. We distinguish Regular-news events from Breaking-news events that trigger dynamic circuit breakers. Using the regression-based jump estimates as efficient benchmarks, we implement the test for fundamental pricing event by event. For Regular-news events, the null of fundamental pricing is rarely rejected, indicating that post-announcement prices typically settle near levels justified by fundamentals. In contrast, Breaking-news events are characterized by substantially higher rejection rates of the null. Pricing errors at these events are predominantly overshooting, suggesting that when dynamic circuit breakers bind, they are associated with systematic deviations from fundamental pricing.

The remainder of this paper proceeds as follows. Section \ref{sec:model} introduces the continuous-time model for the efficient price process and its discrete-time observation model. Section \ref{sec:test} develops the test for fundamental pricing. Section \ref{sec:simulate} examines the finite-sample performance of the testing procedure using a simulation study. Section \ref{sec:apply} tests for fundamental pricing at inflation announcement times when dynamic circuit breakers are triggered after the announcements. Finally, Section \ref{sec:conclude} concludes. Technical derivations and more details about the Breaking-news events are relegated to the Appendix.

\section{The observation model\label{sec:model}}

The efficient log-price of a financial asset is defined on a suitable filtered probability space as
\begin{equation}
X_t=\int_0^t \sigma_s dW_s+\sum_{0 < s \le t}J_s, \qquad t\in(0,1).\label{sm}
\end{equation} 
The continuous component $X_t^{(c)}=\int_0^t \sigma_s dW_s$ is a Brownian martingale driven by a standard Brownian motion $W_t$. The volatility $\sigma_t$ displays standard minimal features (predictability, integrability) such that the stochastic integral is well-defined. The discrete component $X_t^{(d)}=\sum_{0< s \le t} J_s$ captures jumps of size $\Delta J_s=X_s^{(d)}-X_{s-}^{(d)}$. Jumps are triggered by information arrivals, such as macroeconomic announcements. The time interval is normalized to unity; in the empirical application, this corresponds to a fixed local window around a pre-scheduled news release. We make the following assumption on the efficient log-price process \eqref{sm}.

\medskip 

\begin{assump} (VOLATILITY \& JUMPS)\label{A1} 
\bol[(i)]
\li For $h>0$, a path-wise upper bound on integrated variance holds: $\int_t^{t+h}\sigma(s)^2ds\le C_{h}$ 
a.s.\,.
\li $X_t^{(d)}$ is degenerate, triggering a single jump at $t=0.5$ on $(0,1)$.
\eol
\end{assump}

Assumption \ref{A1}(i) imposes a bound on the integrated squared volatility. Assumption \ref{A1}(ii) reflects the structure of pre-scheduled macroeconomic announcements. At time $\tau=0.5$, new information about fundamentals becomes public and is instantaneously incorporated into the efficient price. Importantly, the efficient price fully adjusts at $\tau$. Any delayed  adjustments observed in the transaction price is therefore not attributed to delayed information processing in $X_t$, but rather to frictions in the trading environment.

Traders and econometricians do not observe the efficient price $X_t$ directly. Instead, they observe transaction log-prices $Y_i$ at discrete times $t_i$, $i=1,\ldots,n$, satisfying
\begin{equation}
Y_i=X_{t_i}+\eps_i, \qquad 0 \le t_0 < t_1 < \cdots < t_n < 1. \label{obs}
\end{equation}
We assume equidistant sampling for expositional simplicity, $t_i - t_{i-1} = 1/n$. The additive disturbance $\eps_i$ reflects the fact that the observed prices at financial exchanges may deviate from the latent efficient price because of market frictions and trading rules. We decompose this disturbance into two components:
\[
\eps_i = U_i+H_{t_i}, \qquad i=1,\ldots,n. 
\]
The first component, $U_i$, represents the usual idiosyncratic microstructure noise \citep[see][for example]{hl06, jlmpv09, ax19}. It is defined on an independent product extension of the underlying probability space of the latent process \eqref{sm}. $U_i$ accounts for frictions such as bid-ask bounces and minimum tick size, which appear fully random over time and tend to stochastically dominate the observed price $Y_i$ as $n \rightarrow \infty$. The second component, $H_t$, is an endogenous transition noise process that captures systematic deviations of observed prices from efficient prices during periods when trading rules constrain price adjustments. We impose the following assumption on the two noise components.

\medskip 

\begin{assump} (NOISE) \label{A2} 
\bol[(i)]
\item $U_i \; \overset{i.i.d.}{\sim} \; (0,q^2)$ and has finite fourth moment.
\item $H_t=0$ for $t<\tau$ and $H_t=H_{\bar\tau}$ for $t\ge\bar\tau$.
\eol
\end{assump}

Assumption \ref{A2}(i) is standard in the market microstructure literature, the noise variance is denoted as $q^2$. Assumption \ref{A2}(ii) states that $H_t$ is only active after the event time $\tau$. Once the transition is completed at $\bar\tau$, the deviation remains locally constant at $H_{\bar\tau}$. We do not impose any parametric restriction on the shape of $H_t$ over the transition window $[\tau,\bar\tau)$. This flexibility allows for highly irregular adjustment paths, including interrupted transitions due to temporary trading halts. Crucially, the termination point $\bar\tau$ and hence the transition time $\delta=\bar\tau-\tau$ do not depend on the sampling frequency $n$. This reflects the fact that  trading rules of financial exchanges are defined with respect to fixed look-back windows and therefore do not vanish in more active markets.

\begin{example} \label{ex1} The persistent noise (PN) model of \citet{altz23} parametrizes $H_t$ to capture power-law drifts in the noise at the jump event $\tau$:
\begin{equation}
H_t=-\eta\,J_\tau \left(1-\left(\frac{t-\tau}{\bar\tau-\tau}\right)^{\vartheta} \right)  \mathbbm 1_{\{t\in[\tau, \, \bar\tau]\}}, \qquad \vartheta \in(0,0.5).
\label{pn}
\end{equation}
 Under this specification, the observed price at the event time underreacts by the fraction $\eta \in [0, \,1]$ of the efficient jump. During the transition interval $[\tau, \, \bar\tau)$, the function decays to zero, gradually revealing the jump in the observed price. 
\end{example}

The PN model in Example \ref{ex1} assumes $H_{\bar\tau}=0$, suggesting the convergence of the observed price to the efficient price at the termination time $\bar\tau$. 
For our purposes, we allow for the possibility that $H_{\bar\tau} \neq 0$. This is the case if after the transition period, the observed price settles at a level that is different from the efficient benchmark. This terminal deviation $H_{\bar\tau}$ is the key object of interest to be explored in testing for fundamental pricing. Our main statistical hypothesis to be tested is:

\begin{itemize}
\item[$\H_0:$] $H_{\bar\tau}=0$ \quad (fundamental pricing)
\item[$\H_1:$] $H_{\bar\tau}\neq0$ \quad (non-fundamental pricing) 
\end{itemize}

Under the null hypothesis of fundamental pricing, any deviation between the observed and efficient prices is temporary and restricted to the transition window $[\tau, \, \bar\tau)$. As $t \rightarrow \bar\tau$, the observed price converges to the efficient level implied by $J_\tau$. Under the alternative $\H_1$, instead of pricing the fundamentally justified $J_\tau$, the market either overreacts or underreacts to the news release and settles at a biased price level. This persistent deviation may arise because of additional information effects generated by the trading break itself, or behavioral responses from the market that go beyond the fundamentals of the news release. It is  crucial to note that our econometric methods below do not require parameterization of the transition process. This is because we draw inference from $Y_i$ before $\tau$ and after $\bar\tau$. We do not focus on the actual shape of the transition, but only the termination time $\bar\tau$, whether the observed price has transited to a new price level that is justified by the news about fundamentals.

\section{Test for fundamental pricing\label{sec:test}}

This section develops inferential tools to detect non-fundamental pricing. 
The main objective is to evaluate if the observed price after the transition window settles at a level that is systematically different from the efficient price. We proceed in three steps. First, we show that conventional pre-average jump estimators are inconsistent when the transition window does not vanish asymptotically. We then derive a level-$\alpha$ test for fundamental pricing under the assumptions of known jump size. Lastly, we introduce a regression-based estimator of the jump size and show how it affects the critical value for testing.

\subsection{Why standard jump estimators fail}

The traditional high-frequency literature on jump detection and estimation assumes that price jumps are instantaneous events. In models like \eqref{obs} with microstructure component $\eps_i=U_i$, once a jump occurs, the observed price immediately reflects the magnitude of the jump up to  $U_i$ and a small discretization error. If the sampling frequency $n$ increases, the discretization error vanishes, and $U_i$ can be controlled through smoothing, see the pre-average and spectral estimators in \cite{lm12} and \cite{bnw19}, respectively. These methods allow consistent estimation of the jump size. 

In the presence of exchange-imposed trading limits, however, this logic does not apply. When new information arrives at time $\tau$, prices cannot adjust instantaneously to the new fundamental value, because binding price limits constrain the set of admissible transaction prices. As a result, the market requires a finite time interval of length $\delta= \bar\tau- \tau$ for order flow to gradually incorporate the information and for price to move towards the new equilibrium level step by step. Since the transition time does not vanish as $n \rightarrow \infty$, the discretization error persists asymptotically. The error can be particularly pronounced if volatility also jumps at $\tau$. Hence, the modification changes the asymptotic behavior of conventional jump estimators.

Consider the standard pre-averaging approach described by \citet{jlmpv09} for example. Pre-averaging smooths out the idiosyncratic noise $U_i$ by averaging log-prices over short blocks of length $M_n=c\sqrt n$ for some constant $c>0$. A locally averaged log-price is 
\[
\bar Y_j=M_n^{-1}\sum_{i=j}^{(j+M_n)\wedge n} Y_i,\qquad  j=\lfloor tn\rfloor.
\]
To accommodate a transition interval of length $\delta$, we  define the $\delta$-pre-averaged event return
\begin{equation}
\Delta_{\lfloor \tau n\rfloor,\delta}^{n} \hat X=\bar Y_{\lfloor (\tau+\delta) n\rfloor }-\bar Y_{\lfloor \tau n\rfloor-M_n}.\label{par}
\end{equation}
When $\delta=0$, this reduces to the standard pre-average jump estimator of \citet{lm12}. When $\delta=(\bar\tau-\tau)>0$, the estimator compares the average price before the event time $\tau$ with the average price after the termination time $\bar\tau$. The following lemma is an extension of Lemma 1 of \cite{lm12} and a direct translation of Proposition 3.1 of \cite{bnw19}, characterizing the asymptotic behavior of \eqref{par}.

\newpage

\medskip 

\begin{lemma} \label{l1}
Consider model \eqref{obs} under Assumptions \ref{A1} and \ref{A2} with $\delta= \bar\tau -\tau$. Then as $n\rightarrow\infty$, the $\delta$-pre-average return \eqref{par} over the transition interval $[\tau, \, \bar\tau)$ obeys
\begin{equation}
n^{1/4}\left( \Delta_{\lfloor \tau n\rfloor,\delta}^n\hat X-\big(J_\tau+\text{BIAS}_{\delta}\big)\right)\overset{st}{\longrightarrow} \eta_\delta, \label{clt}
\end{equation}
where 
\[
\text{BIAS}_\delta=H_{\bar\tau}+\int_{\tau}^{\bar\tau-}\sigma_sdW_s, 
\]
and $\eta_\delta \sim \mcl{N} \left(0, V_{\tau,\bar\tau} \right)$ with variance $ V_{\tau,\bar\tau} =\frac{1}{3}(\sigma^2_{\tau-}+\sigma^2_{\bar\tau})c +2q^2.$
\end{lemma}

Lemma \ref{l1} states that accounting for a non-vanishing transition time of $\delta>0$ introduces the term $\textit{BIAS}_\delta$. It consists of the term $H_{\bar\tau}$ which captures potential non-fundamental pricing and the increment of the efficient price \eqref{sm} (Brownian martingale) over $\delta$. The Brownian component is analogous to a discretization error studied in the high-frequency literature. However, discretization errors vanish because sampling windows shrink to zero as $n\rightarrow\infty$. Here, by contrast, the transition window is fixed in time, reflecting the fixed look-back windows of exchange trading rules. Therefore, the Brownian increment accumulated over $[\tau, \, \bar\tau)$ remains even as sampling frequency increases. As a result, the $\delta$-pre-average estimator is not a consistent estimator of the jump $J_\tau$.

\subsection{Testing fundamentalness\label{subsec:iftest}}

While Lemma \ref{l1} suggests that we cannot directly estimate the jump $J_\tau$ with simple pre-averaging methods, it provides the terms to control for in order to test for fundamental pricing $\H_0$. We start by assuming $J_\tau$ to be known. Consider the test statistic $\lvert\Delta_{\lfloor \tau n\rfloor,\delta}^n \hat X-J_\tau\rvert$, which measures the deviation of the $\delta$-pre-averaged return from the fundamentally justified true jump size. According to Lemma \ref{l1}, the test statistic consists of the following three terms: $H_{\bar\tau}+\int_{\tau}^{\bar\tau-}\sigma_sdW_s+\eta_\delta $. Under $\H_0$, $H_{\bar\tau}=0$, the test statistic reduces to the Brownian martingale increment over the transition interval and the $\delta$-pre-average estimation error. Under $\H_1$, $H_{\bar\tau} \neq 0$, a non-zero pricing error remains after the transition is completed.

\medskip 

\begin{prop}\label{p1}
Given Lemma \ref{l1}, under $\H_0$ and a small test level $\alpha$, the test function 
\[
\varphi_{\alpha,\delta}=\mathbbm 1 { \left\{ \Big| \Delta_{\lfloor \tau n\rfloor,\delta}^n\hat X-J_\tau \Big|>\kappa_{\alpha,\delta} \right\}}
\] 
with critical value
\begin{equation}
\kappa_{\alpha,\delta} = \left( \sqrt{C_\delta} + n^{-1/4}\sqrt{V_{\tau,\bar{\tau}}} \right) \sqrt{\log\!\left(\frac{8}{\pi\alpha^2}\right) - \log\log\!\left(\frac{8}{\pi\alpha^2}\right)}
\label{kappa}
\end{equation}
 provides an asymptotic level-$\alpha$ test:
 $\P(\varphi_{\alpha,\delta}=1)\leq \alpha.
 $
\end{prop}

The critical value $\kappa_{\alpha,\delta}$ in Proposition \ref{p1} uses a union bound on the Brownian fluctuation $\int_\tau^{\bar\tau-}\sigma_sdW_s$ and the $\delta$-pre-average estimation error $\eta_\delta$. Both components refer to a Gaussian tail bound, which is stated up to second order, explaining the log terms and depending on the joint error rate $\alpha/2$. The absolute increment of the Brownian martingale over $\delta$ is bound by the worst-case scenario using $C_\delta$ from Assumption \ref{A1}(i). This bound increases with the length of the transition window $\delta$. In other words, the critical value becomes larger when longer transition interval is used, which makes the test of non-fundamental pricing less sensitive. The $\delta$-pre-average estimation error is negligible in high-frequency observations with large $n$. The test is an asymptotic test because it builds on the Central Limit Theorem in Lemma \ref{l1}. It is conservative because of the bounding of integrated squared volatility by $C_\delta$ and the Gaussian tail approximation. The tail approximation becomes tight for $\alpha\rightarrow 0$.

The test in Proposition \ref{p1} is infeasible in practice, because the true jump size $J_\tau$ is unknown. We develop a feasible test below when the true jump size can be consistently estimated.

\subsection{Feasible test with estimated $J_\tau$ \label{subsec:ftest}}

We exploit cross-event variation to enable consistent estimation of $J_\tau$. For a single event, the Brownian martingale increment over the transition window contaminates the estimation of the jump. However, when averaging across multiple events, the Brownian component cancels out, allowing the jump component to be identified. 

\newpage

\medskip 

\begin{assump} (JUMP FACTOR) \label{A3} Across a homogeneous set of events, jump sizes in the efficient log-prices satisfy
\[  J_\tau =  F \,b, \]
where $F$ is a $K$-dimensional observable news factor, and $b$ is a fix $K$-vector of loadings.
\end{assump}

Assumption \ref{A3} reflects the asset-pricing view that a jump in the efficient price is a linear function of the information content of the news release. 
For a homogeneous set of $N$ events, this provides sufficient structure to estimate the loading $b$. We apply a leave-one-out approach to regress the $\delta$-pre-averaged returns on the news factors. The left-out event is the one to be tested for fundamental pricing, while $N$ events that are assumed to be predominantly fundamentally priced are used to estimate $b$. Denote the associated $N$-vector of $\delta$-pre-average returns by ${ \Delta_{\lfloor \tau n\rfloor,\delta}^{n}\hat  X}^{*}$ and the corresponding $(N\times K)$ matrix of factors $F^*$. Without loss of generality, the factors $F^*$ are normalized to have unit variance. The loading $b$ is estimated via the regression
\begin{equation}
    \Delta_{\lfloor \tau n\rfloor,\delta}^{n}\hat X^{*} = F^{*} b + e^{*}.
    \label{reg1}
\end{equation}
Combining Assumption \ref{A3} with the decomposition in Lemma \ref{l1}, the $N$-vector $e^*$ has elements 
\begin{equation}
    e^{*(j)} = \int_{\tau}^{\bar\tau-}\sigma_s^{(j)} dW_s +\eta_{\delta}^{(j)}, \qquad j=1,...,N,
    \label{error}
\end{equation}
which shows that the error in the linear regression \eqref{reg1} is heteroskedastic. Therefore, if $F^{*}$ is strongly exogenous, the least square estimator $\hat b$ is unbiased and consistent as $N \rightarrow \infty$. That is, via the classical regression approach, the conditional expectation at the tested event $F\hat b$ provides a consistent estimator of the jump in the efficient price, $J_\tau$. 
Replacing $J_\tau$ in Proposition \ref{p1} by $F\hat b $ leads to a feasible test for fundamental pricing. 
 
\medskip 

\begin{prop}\label{p2}
Under the conditions of Proposition \ref{p1} and Assumption \ref{A3}, and using the least square estimator of $b$ from \eqref{reg1} with $\Big|e^{*(j)} \Big|\le c_{e^*}$ for $j=1,...,N$. Under $\H_0$, the test function
\[
\varphi_{\alpha,\delta}^*=\mathbbm 1{ \left\{ \Big| \Delta_{\lfloor \tau n\rfloor,\delta}^n\hat X-F\hat b  \Big| >\kappa_{\alpha,\delta}^* \right\}}
\] 
with critical value 
\begin{equation}
 \kappa_{\alpha,\delta}^* = \kappa_{2/3\alpha,\delta} +N^{-1/2} c_{e^*} \|F\|_1\sqrt{ 2 \log\left(\frac{6}{\alpha}\right)}
\label{kappa_tilde}
\end{equation}
provides an asymptotic level-$\alpha$ test: $\P(\varphi_{\alpha,\delta}^*=1)\leq \alpha$.
\end{prop}

The adjusted critical value $\kappa_{\alpha,\delta}^*$ augments the bound from Proposition \ref{p1} to account for the estimation error in estimating $b$. The correction term decreases as the number of events $N$ used to estimate $b$ increases. The factor vector $F$, containing the $K$ observable components at the tested event, enters through its $\ell_1$-norm. The log term arises from the sub-Gaussian tail bound for $e^*$, while the numerical constants follow from a union bound over the three components of the test statistic. We rely on Hoeffding's inequality to bound the sum of regression errors. Therefore, we do not require specific distributional assumptions for $e^*$ except that the elements of $e^*$ are centered, independent, and uniformly bounded by $c_{e^*}$. 
Capturing data availability in our empirical study, the test is asymptotic in $n$ but non-asymptotic in cross-sectional events $N$.

Under $\H_1$, the non-zero $H_{\bar\tau}$ enters the test statistic through the $\delta$-pre-average return $\Delta_{\lfloor \tau n\rfloor,\delta}^n\hat X$. Because all fundamental information is captured by $F$ and incorporated through $F\hat b$, while estimation errors and Brownian fluctuation are controlled by $\kappa_{\alpha,\delta}^*$, rejections of $\H_0$ can be interpreted as evidence of non-fundamental pricing. The rejections become more likely as the magnitude of the mispricing gets larger. In the simulation study below, we evaluate how large $H_{\bar\tau}$ needs to be for non-fundamental pricing to be detected with high probability.  

\begin{remark}
The testing procedure developed above requires knowledge of the termination time $\bar\tau$. 
In practice, however, $\bar\tau$ is not observed and needs to be inferred from the data. When the transition component fades out slowly, such as the power-law dynamics in the PN model in Example \ref{ex1} with small $\vartheta$, it can be difficult to determine when the transition has effectively ended. This presents a practical challenge that we abstract from in the present analysis. Our approach is to discuss different choices of $\delta$ while implementing the testing procedure below. Proposition \ref{p1} suggests selecting $\delta$ too large does not invalidate the test, but it reduces the power. This is because a larger $\delta$ inflates the critical value $\kappa_{\alpha,\delta}$ through the bound on integrated squared volatility $C_\delta$, making rejections less likely.
\end{remark}

\section{Simulation study\label{sec:simulate}}

The aim of the simulation study is to evaluate the finite-sample performance of the proposed testing procedure from Section \ref{sec:test}. We first demonstrate the estimation bias in the $\delta$-pre-average jump estimator and how the regression-based estimator reduces this bias. Then we investigate the size and power properties of the test for fundamental pricing. 

\subsection{Simulation setup}

We simulate the efficient log-price $X_t$ using a Heston-type stochastic volatility model with a single jump in price and volatility at the announcement time:
\begin{eqnarray}
X_t &=& \int_0^t \sigma_s d W_s +J_\tau \mathbbm 1_{ \{ t \geq \tau \}}, \label{xt} \\
d\sigma_t^2 &=& \kappa \left(\theta -\sigma_t^2\right) dt+ \varsigma\, \sigma_t dB_t+ J_{\sigma(\tau)}\,e^{-\kappa_J(t-\tau)}\mathbbm 1_{\{t\ge \tau\}}, \quad t\in[0,\,T], \label{hv}
\end{eqnarray}
where $\tau =0.5 T$ denotes the news-release time. The time interval $[0,\,T]$ represents a 3-hour window centered around the event time $\tau$. The standard Brownian motions $W_t$ and $B_t$ have correlation $E (dW_t\, dB_t) =\rho dt$. At the event time, $X_t$ jumps by $J_\tau = 0.3F$, where the factor $F$ is drawn independently across simulations from $\mcl{N} (0, \, 0.03^2)$. The annualized parameters of the squared volatility are set to $(\kappa, \,\theta, \,\varsigma, \,\rho)=(5, \,0.0225, \,0.4, \,-\sqrt{0.5})$, which are standard choices in the high-frequency literature. 
Empirical evidence suggests that macroeconomic news announcements often coincide with temporary shifts in volatility.\footnote{Methods for assessing volatility jumps are discussed by \cite{tt11} and \cite{bw18}, for example.} To capture this feature, we augment the Heston volatility process \eqref{hv} with a volatility jump at the announcement time. The magnitude of the volatility jump is $J_{\sigma(\tau)}=c_J\theta$. We draw $c_J$ from a uniform distribution over $[4,16]$ such that the jump in volatility is between 2 and 4 times of the pre-event mean $\theta$. The decay rate is set to $\kappa_J=2,\!500$ such that the half-life of the volatility jump is approximately 45 minutes. 

The observed log-price is in accordance with \eqref{obs}, a discretized and noisy superposition of the efficient log-price. We set $n=21,\!600$ with equally spaced observations over $[0, \,T]$. This corresponds to 2 observations per second and matches the minimum number of trades observed in our S\&P 500 futures data over the 3-hour interval. The idiosyncratic noise $U_i$ is drawn from $\mcl{N} (0, \, q^2)$. We set $q= 0.004\%$ to match the maximal noise level observed across all events in the empirical data. Although our methodology does not require specifying the transition noise $H_t$ for $t\in[\tau,\bar\tau)$, we also examine scenarios in which the pre-average returns is constructed using a value of $\delta$ that is too small. In this case, the termination time $\bar\tau$ is chosen below its true value. Under such misspecification, it becomes necessary to explicitly model $H_t$. To this end, we consider the PN model of \citet{altz23}, as described in Example \ref{ex1}, setting $\eta=1$. The true termination time is fixed at $\bar\tau=30$ seconds.

\begin{figure}[t]
\caption{Simulated price path\label{fig:ex1}}
\includegraphics[width=0.49\textwidth]{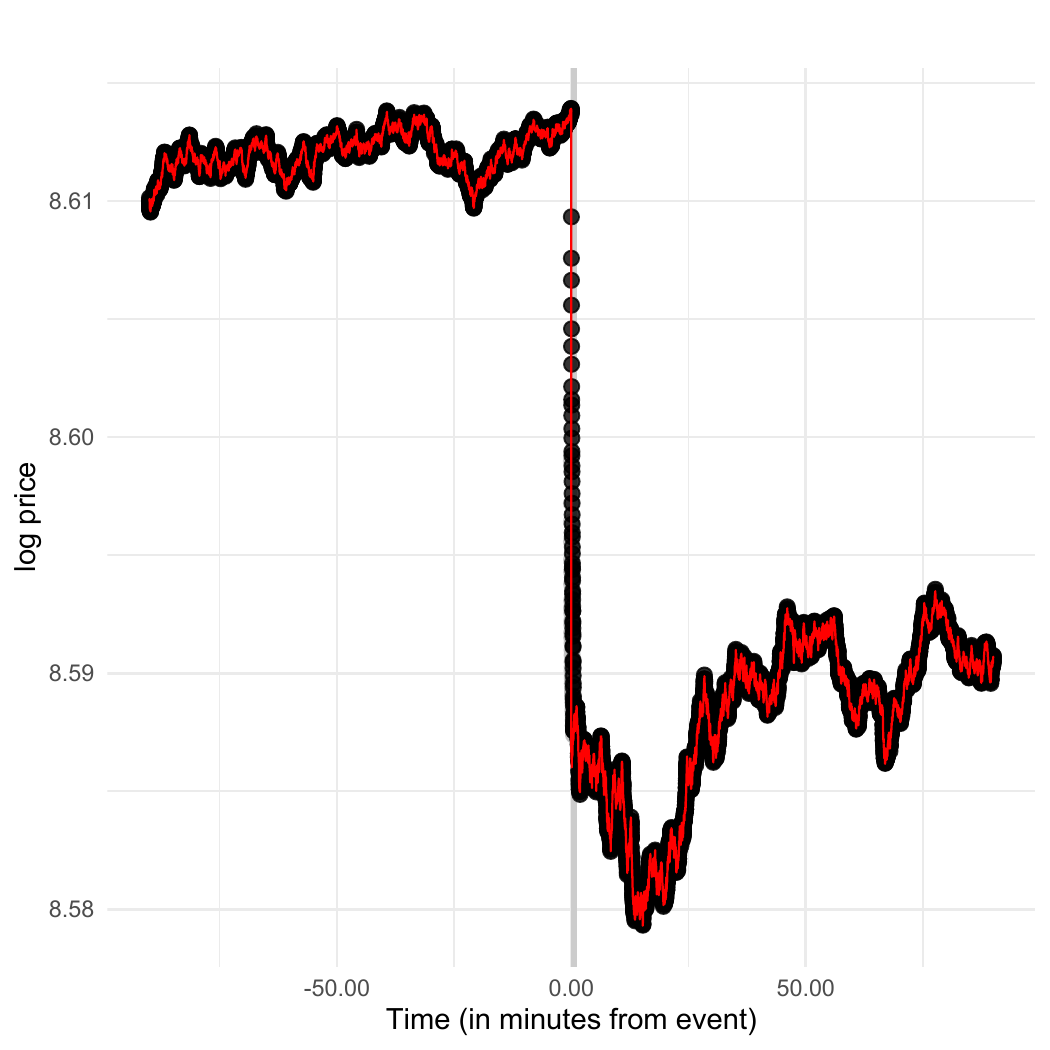}\hspace{0.1cm}
\includegraphics[width=0.49\textwidth]{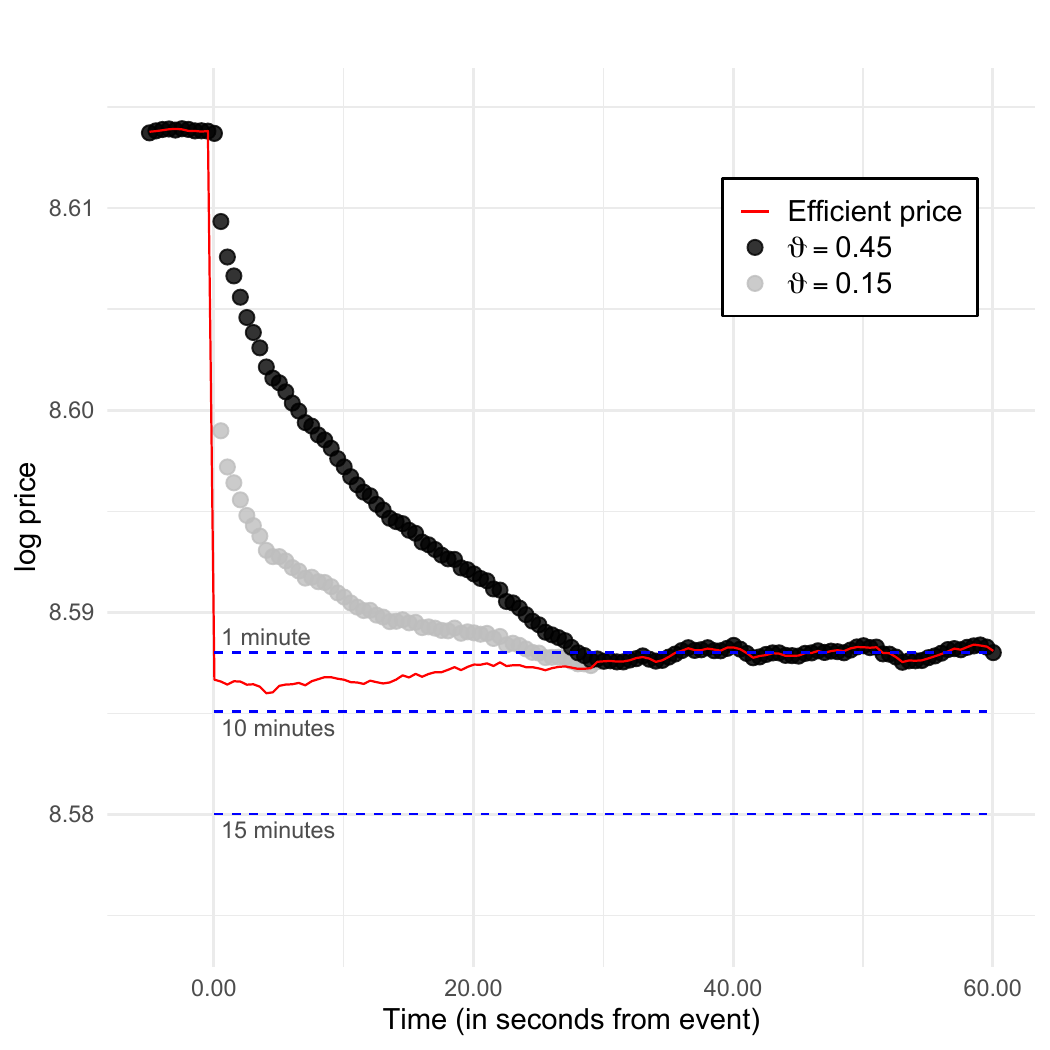} \vspace{0.5cm}

\footnotesize
\parbox{17 cm}{\emph{Note:} Left: Price path over a 3 hour interval with shaded 1 minute interval after the jump event. Right: Same simulated path but zoomed at event time. Blue horizontal dashed lines indicate the efficient price 5, 10 and 15 minutes after the event. 
}
\end{figure}

Figure \ref{fig:ex1} depicts an example of a simulated price path to illustrate the characteristics of the simulated processes. The left panel shows the full 3-hour window around the announcement time $\tau$, while the right panel zooms into the interval from right before the jump to 60 seconds after the jump. The red line represents the efficient log-price $X_t$, which displays a large instantaneous jump of size $J_\tau = -2.7\%$. The black and gray dots on the right panel correspond to the observed log-prices $Y_i$ under the PN model with $\vartheta=0.45$ and 0.15, respectively. In both specifications, the observed prices do not adjust instantaneously to the new efficient level. Instead, the jump in $X_t$ is initially offset by the transition component $H_t$, and the adjustment unfolds gradually over the interval $[\tau , \,\bar\tau)$. It is evident from Figure \ref{fig:ex1} that the smaller value of $\vartheta$ implies a flatter fade-out near $\bar\tau$. Consequently, when $\vartheta$ is small, choosing a value of $\delta$ that is too small is less detrimental for pre-average return than in cases where the decay is steeper near $\bar\tau$. In the analysis that follows, we focus on the case of $\vartheta=0.45$.

The blue dashed horizontal lines in the right panel of Figure \ref{fig:ex1} indicate the efficient price evaluated 1, 10, and 15 minutes after the announcement. The distance between the actual jump at time $\tau$ (red line) and the dashed blue lines represents the associated $\text{BIAS}_\delta$ from Lemma \ref{l1}. Using a 1-minute window would understate the true efficient jump by 1.6\%. If a 15-minute window is used, we would overstate the true jump by 28\%. This discrepancy directly reflects the Brownian component of the bias term $\text{BIAS}_\delta$ in Lemma \ref{l1}, which tends to increase with the length of the transition window.

The tests for fundamental pricing in Propositions \ref{p1} and \ref{p2} require an estimate of volatility. Estimating integrated volatility over intervals with a strong drift component is challenging. \cite{lrs26} analyze settings and propose estimators that address this issue. Because our framework allows for potentially irregular adjustments, such as cases in which a circuit breaker interrupts a transition, we do not impose a specific transition structure. To obtain the bound $C_{\delta}$, we estimate integrated squared volatility using 30-second windows over the full 3-hour interval. With $\text{IV}_{30}^{\text{pre}}$ the collection of 30-second pre-average realized volatility prior to $\tau$, we set $ C_\delta = \delta c_\Delta \max(\text{IV}_{30}^{\text{pre}})/{30}$,
where the scaling factor $c_\Delta$ is computed as the ratio of average $IV_{30}$'s after $\bar\tau$ and before $\tau$.

\subsection{Assessment of price-jump estimators}

We evaluate the finite-sample behavior of two estimators. The first is the $\delta$-pre-average estimator $\Delta_{\lfloor \tau n\rfloor,\hat\delta}^{n}\hat X$, with ad-hoc fixed transition window length $\delta$ ranging from 20 seconds to 10 minutes. We compare it with the regression-based estimator $\hat{J}_\tau = \hat{b} F$, which exploits the cross-event factor structure. In the simulation, we use $N=10$ and $N=50$ events to estimate the loading $b$ using least square method. Table \ref{Tje} reports the mean absolute percentage error (MAPE) and the mean squared error (MSE) across different estimation designs. Several findings emerge. 

\begin{table}[t]
\centering
\begin{threeparttable}
\caption{Evaluation of jump size estimation error \label{Tje}}
\begin{tabular}{@{\hskip 0.5cm}L{4cm}C{1.5cm}C{1.5cm}C{1.5cm}C{1.5cm}C{1.5cm}C{1.5cm}@{\hskip 0.5cm}}
\toprule \hline
&\multicolumn{6}{c}{Length of $\delta$ window (seconds)}\T\B \\ \cmidrule(lr){2-7} 
& 20 &30&40&60&300&600\T\B \\ \hline
\bf $\delta$-pre-average\T\B \\ \cline{1-1}
MAPE       &0.216&0.110&0.127&0.156&0.341&0.465\T \\
MSE	&  0.101& 0.104& 0.142&0.226&1.064&1.992\T\B \\  \hline
\bf Regression ($N=10$)\T\B \\ \cline{1-1}
MAPE        & 0.539& 0.032& 0.032& 0.040& 0.088& 0.121\T \\
MSE	&2.717&0.015&0.012&0.019&0.094&0.196\T\B \\  \hline
\bf Regression ($N=50$)\T\B \\ \cline{1-1}
MAPE       & 0.531& 0.018&0.015&0.017&0.037&0.055\T \\
MSE &2.417&0.003&0.003&0.005&0.018&0.035\T\B \\  \hline
\bottomrule
\end{tabular}
\begin{tablenotes}[flushleft]
\footnotesize
\item Notes: MAPE denotes the average absolute percentage error. Based on 25,{0}00 simulation rounds with $n=21,\!600$ high-frequency observations. $N$ is the number of cross-events for the regression-based jump estimator.
\end{tablenotes}
\end{threeparttable}
\end{table}

First, when the pre-averaging window $\delta$ is at least as large as the true transition length of 30 seconds, the estimation error increases with $\delta$. For the simple $\delta$-pre-average estimator, the MAPE rises from 11\% for 30-second window to 46.5\% for 10-minutes window. With a 10-minute window, the pre-averaged event return exceeds the true jump size by almost 50\% on average. In line with the decomposition in Lemma \ref{l1}, the increase in the estimation error for longer windows is due to the Brownian variation in the efficient price over the $\delta$ window.

Second, incorporating cross-event information through the regression-based estimator substantially reduces the estimation error. Averaging over $N = 10$ events lowers MAPE considerably, especially for larger $\delta$. When $\delta$ is larger than 30 seconds, MAPE declines by more than 70\% compared to the $\delta$-pre-average estimator. Moreover, increasing the number of events used to estimate the factor loadings further improves accuracy. Raising $N$ from 10 to 50 reduces MAPE by another 50\%. This result highlights the gains from cross-sectional averaging.

Finally, if $\delta$ is chosen smaller than the true transition length, the estimation error becomes large and cannot be mitigated by pooling across events. In our simulation, selecting a $\delta$ that understates the true transition time by 10 seconds produces a MAPE of 21.6\% for the $\delta$-pre-average estimator. This error cannot be corrected by the regression approach. Indeed, the regression estimates of the jump increase MAPE to 53\%, because the price transition is not completed using 20-second window.

\subsection{Assessment of the test for fundamental pricing}

We now proceed to assess the size and power properties of the test for fundamental pricing from Proposition \ref{p2}. We focus on the regression-based jump estimator with $N=50$ and set the test level to $\alpha =0.01$. Figure \ref{Fpower} illustrates the rejection frequencies of $\H_0$ as a function of the relative magnitude of the pricing error $H_{\bar\tau}$. Zero on the horizontal axis corresponds to the null hypothesis $H_{\bar\tau}=0$. Under $\H_1$, we increase $H_{\bar\tau}$ to up to 43\% of the average absolute jump size $J_\tau$. 

\begin{figure}[t]
\centering
\caption{Size and power for testing fundamental pricing\label{Fpower}}
\includegraphics[width=0.6\textwidth]{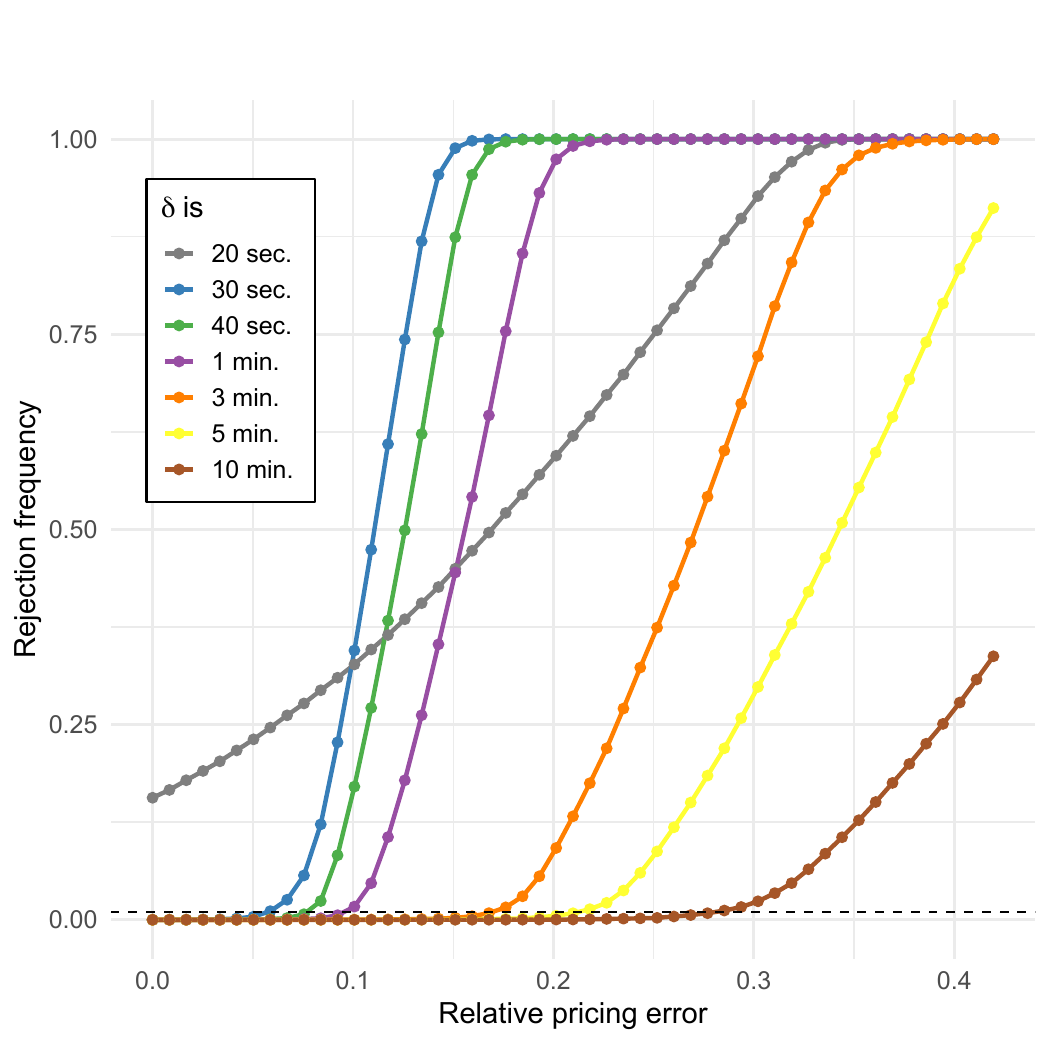}\vspace{0.2cm}

\footnotesize
\parbox{15 cm}{\emph{Note:} Results are based on 25,000 simulations and the critical value from Proposition \ref{p2}. Jump sizes are estimated using the regression-based estimator $\hat bF$. $\delta$ indicates the window length for the pre-average event-return estimation. The true transition time is 30 seconds. The $x$-axis display the size of the pricing error at termination $H_{\bar\tau}$ relative to the average jump size across all simulations.}
\end{figure}

Figure \ref{Fpower} confirms that under $\H_0$, when $\delta$ is at least as large as the true transition length of 30 seconds, the rejection frequencies of $\H_0$ are below the nominal level of the test. In particular, false rejections fall below 1\% for larger $\delta$ under $\H_0$. If $\delta$ is shorter than the true transition and the remaining adjustment is steep near $\bar{\tau}$, the jump estimator becomes biased because the transition has not fully decayed. The resulting estimation error documented in Table \ref{Tje} propagates through to the test statistic and induces false rejections. In our simulations, using estimation windows that undercut the true transition time by 10 seconds leads to a rejection rate of about 16\% under $\H_0$.

Under the alternative $\H_1$, the rejection rates quickly increase towards one. For window sizes that are at least 30-second long, smaller $\delta$ yields a quicker transition of the power curve. This is because the Brownian martingale component accumulates over the window length $\delta$, which inflates the critical value when $\delta$ is larger and hence makes the test less sensitive. For mispricing of about 15\%, tests based on estimation windows between 30 and 60 seconds display between 40\% to 100\% correct detections. In contrast, tests based on longer windows of 3 to 10 minutes can only achieve roughly 1\% detection rates for the same sized misspricing of 15\%. These results suggest that accurate localization of the transition time may enhance power by reducing the influence of the Brownian variability on the critical value. Excessively long windows dilute the test statistic and make the test of fundamental pricing overly conservative.

\section{Empirical application\label{sec:apply}}

We apply the proposed testing framework to front-month E-mini S\&P 500 futures traded at the CME. Our objective is to evaluate whether triggering events of dynamic circuit breakers affect the pricing of macroeconomic fundamentals. 

\subsection{Data description}

We focus on pre-scheduled U.S.\ CPI inflation releases at 8:30 a.m.\ Eastern Time (ET) between January 2020 and October 2025. These events provide a homogeneous and economically meaningful set of announcements that  regularly generate significant price movements. Importantly, a non-negligible fraction of these releases trigger trading breaks shortly after the release time. We classify \textit{Breaking news} (BN) as events for which a trading break of at least 4.99 seconds is triggered within the 3-minute window from 8:30 a.m.. The remaining announcements are referred to as \textit{Regular news} (RN). Among the 69 inflation releases in our sample, 14 ($\sim$ 20\%) are classified as Breaking news. The complete list of Breaking-news dates is provided in Table \ref{tab:breaks}. In all 14 cases, the time between the news release and the first treading break is at most 2 seconds. The median duration of the trading halt is 5 seconds, and the maximum is 20 seconds.

The descriptive statistics reported in Table \ref{T:tradestat} suggest that Breaking news is different from Regular news. Trading activity almost doubles from a median of roughly 50 to 80 trades per second. The median 1-minute absolute event return quadruples from 0.2\% to 0.8\%. The largest market response in the S\&P 500 futures over an event window of 1 minute is 1\% for Regular news and 2.5\% for Breaking news. In contrast to the event returns, the difference in volatility between Breaking news and Regular news is smaller. However, volatility increases by a factor that is larger than two after the news release, suggesting the possibility of volatility jumps at the announcement time.

\begin{table}[t]
\centering
\begin{threeparttable}
\caption{Breaking news vs. Regular news \label{T:tradestat}}
\begin{tabular}{@{\hskip 0.5cm}l@{\hskip 0.5cm} C{1cm}C{1cm} C{1cm}C{1cm} C{1cm}C{1cm} C{1cm}C{1cm}C{1cm}C{1cm}@{\hskip 0.5cm}}
\toprule \hline
 & \multicolumn{2}{c}{\multirow{2}{*}{\makecell{Trades per\\  second}} }&
\multicolumn{2}{c}{\multirow{2}{*}{\makecell{Event return \\ size (\%)}} }&
\multicolumn{2}{c}{\multirow{2}{*}{\makecell{Volatility \\ shift (factor)}} } &
\multicolumn{4}{c}{Factor $F$}\T\B \\
&&&&&&&\multicolumn{2}{c}{Surprise}&\multicolumn{2}{c}{Attention}\B \\ \cmidrule(lr){2-3} \cmidrule(lr){4-5} \cmidrule(lr){6-7} \cmidrule(lr){8-9} \cmidrule(lr){8-9}\cmidrule(lr){10-11}
&RN&BN&RN&BN&RN&BN&RN&BN&RN&BN\\ \midrule
Median&47.9&78.8&0.19&0.82& 2.65&2.92&0.07 &0.20&48.0&55.0\T \\
Q25\%&17.2&57.0& 0.06&0.54&1.74&2.00&0.04&0.08&30.0&47.0 \\
Q75\%&77.6&107&0.43&1.61&3.29&3.19&0.10&0.21&54.0&74.0 \\
MAD&45.7&37.7&0.23&0.48&1.03&0.91&0.04&0.09&16.3&16.3\\
Min&3.32&32.6&0.01&0.11&0.09&0.60&0.00&0.00&20.0&39.0\\
Max&130&157&1.05&2.49&5.95&3.89&0.56&0.26&73.0&100\B \\ \hline
\bottomrule
\end{tabular}
\begin{tablenotes}[flushleft]
\footnotesize
\item Notes: `RN' includes 55 Regular-news events, `BN' represents 14 Breaking-news events. The event time $\tau$ is 8:30 a.m. ET CPI inflation news release time. Trades per second is calculated based on the 1-minute interval between 8:30 and 8:31. Absolute event returns are estimated using one-minute pre-average returns. Volatility shift is the square root of integrated variance before $\tau$ and after 8:31. Surprise is Bloomberg mean surprise in percentage points. Sample period is from January 2020 to October 2025.
\end{tablenotes}
\end{threeparttable}
\end{table}

We use two factors to quantify the fundamental content of the news releases. The first factor is the surprise component of the news release. The surprise is the difference between the announced inflation and the mean of Bloomberg economists' survey forecasts. The second factor is a monthly Google attention index based on the search query ``inflation". 
The last four columns of Table \ref{T:tradestat} display the difference in Bloomberg-survey surprises and the Google index between Regular-news and Breaking-news events. The mean month-on-month surprise across survey participants has a median across events of 7 basis points on Regular news days and 20 basis points on Breaking-news days. Thus, the surprise is almost 3 times larger on Breaking-news days. Attention also varies with slightly higher attention at Breaking-news events.

\begin{figure}[p]
\centering
\caption{Examples of S\&P 500 futures paths around inflation release times\label{F:ts}}
\subfloat[Regular news (May 15, 2024)]
{\includegraphics[width=0.45\textwidth]{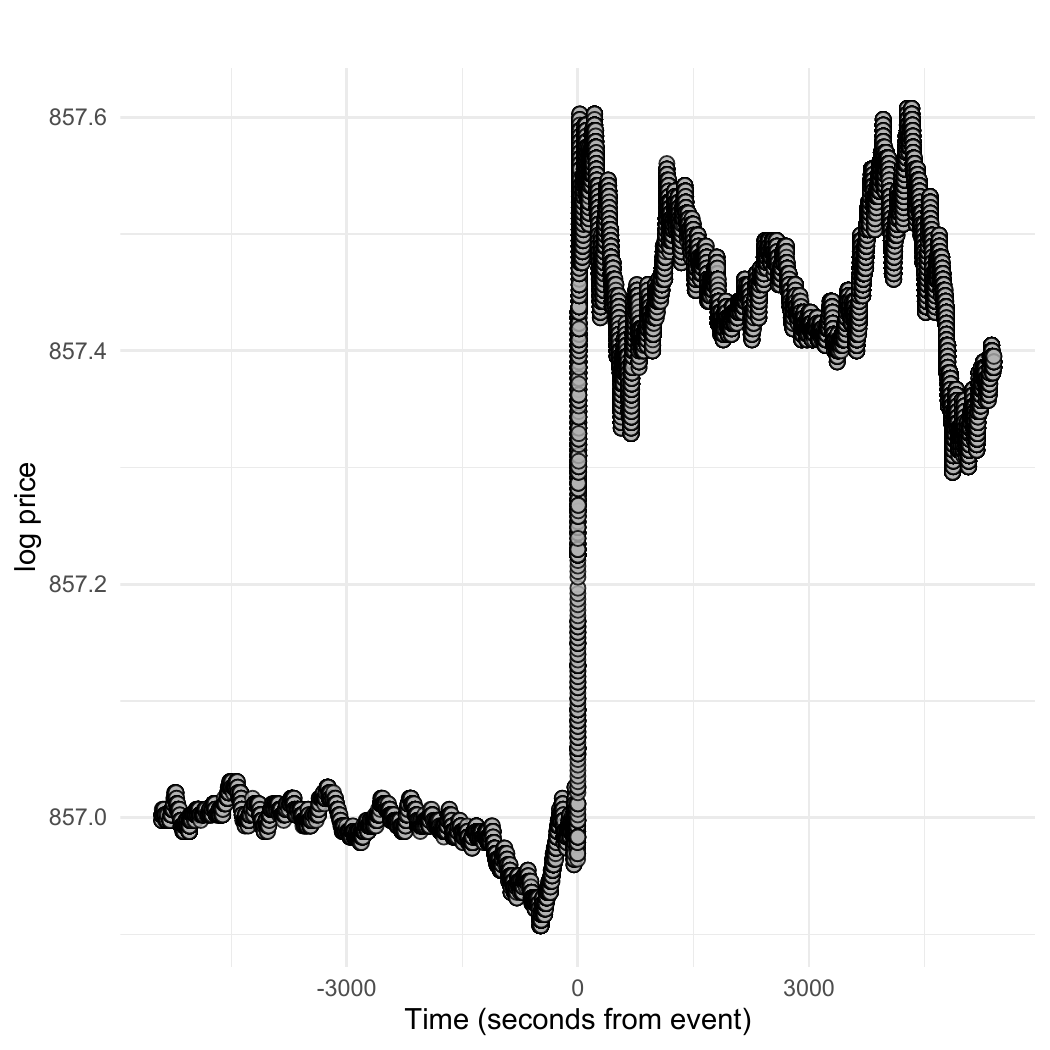}\hspace{0.25cm}
\includegraphics[width=0.45\textwidth]{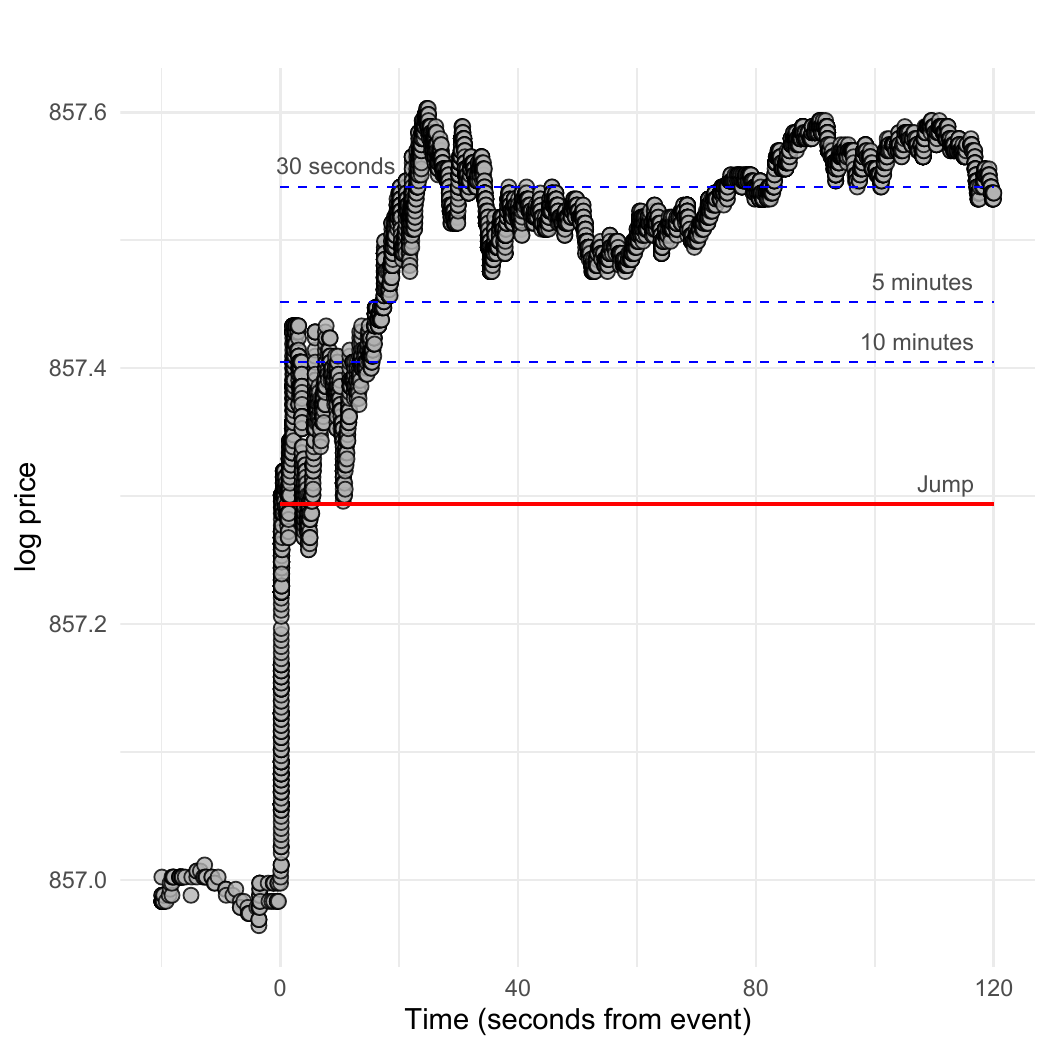}} \\
\subfloat[Breaking news (February 12, 2025)]
{\includegraphics[width=0.45\textwidth]{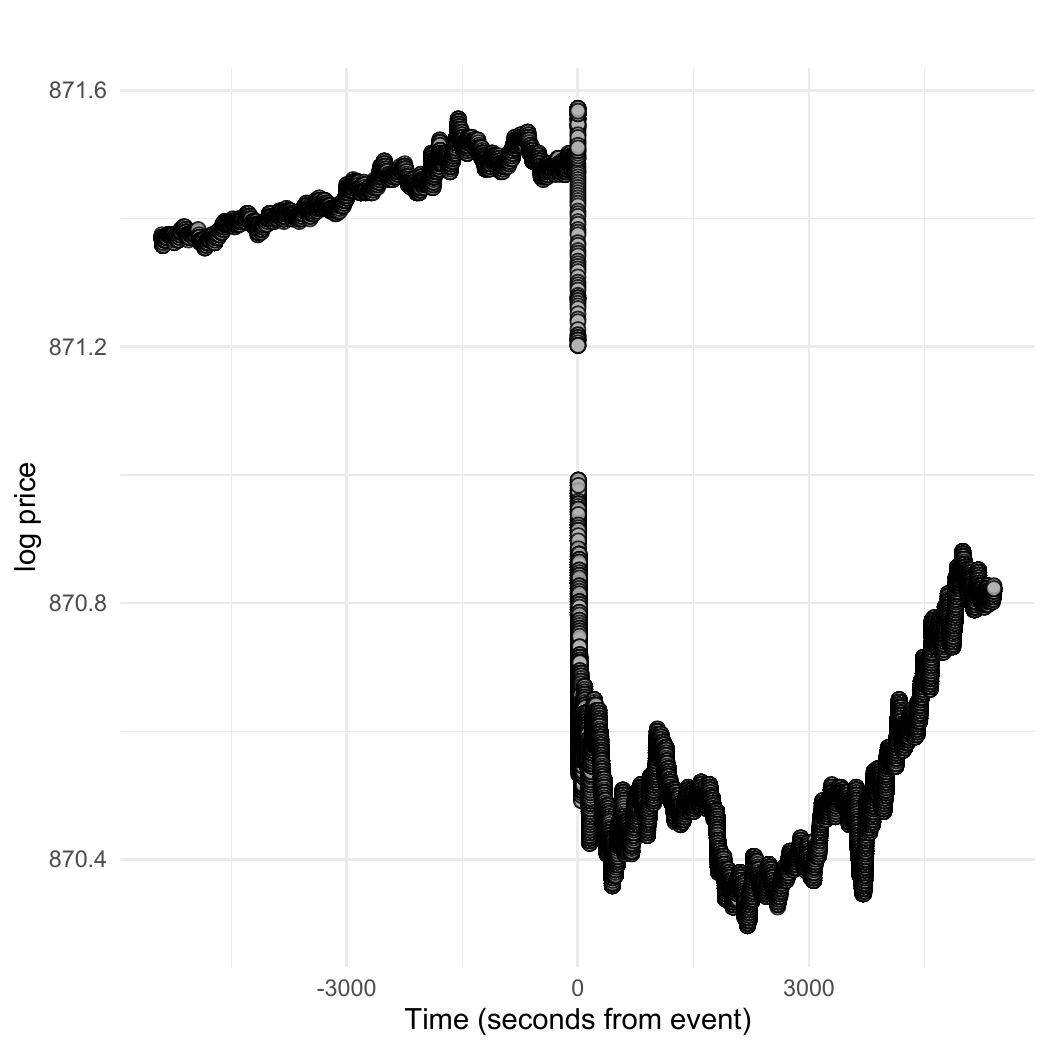}\hspace{0.25cm}
\includegraphics[width=0.45\textwidth]{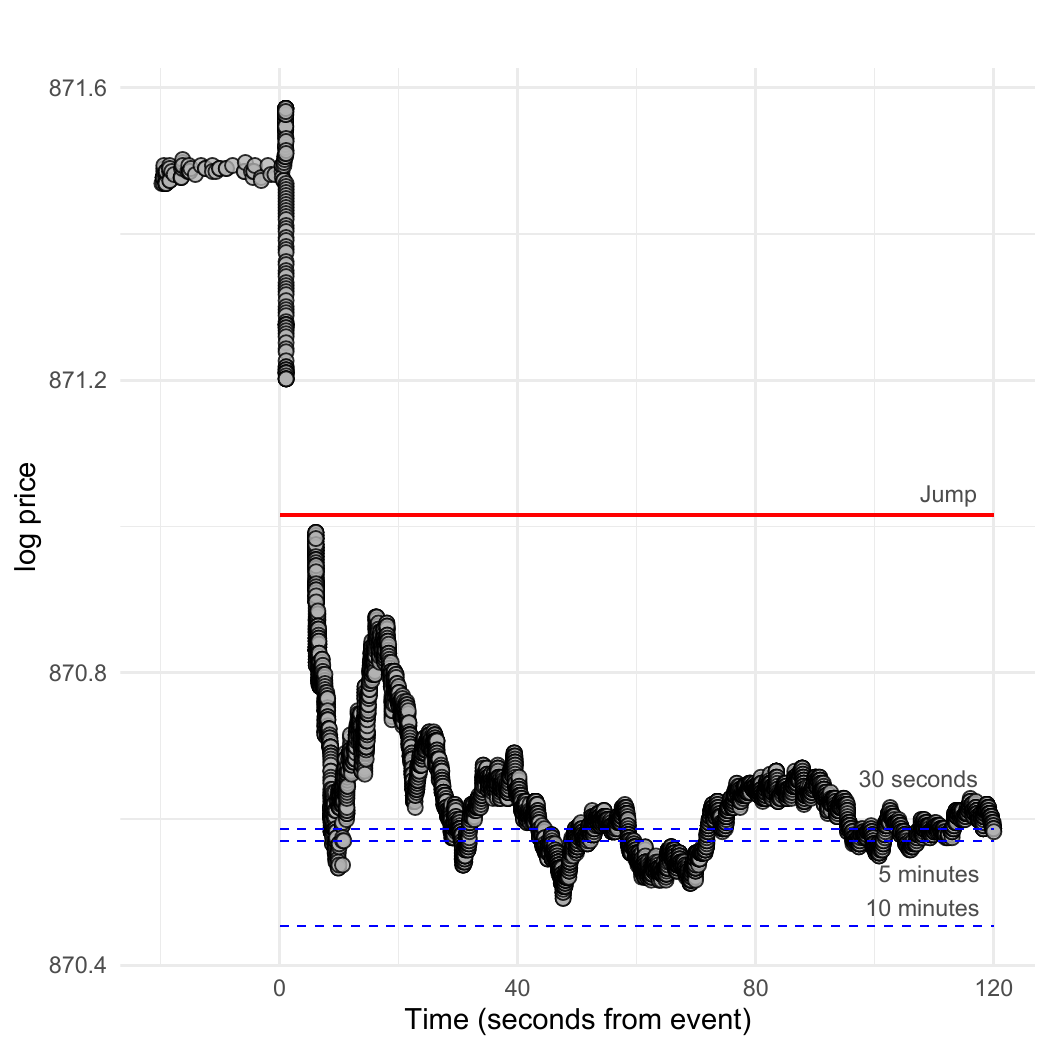}} \\
\footnotesize
\parbox{16 cm}{\emph{Note:} Examples for 8:30 a.m.\ news pricing (time=0). Left column: 7:00 a.m.\ to 10:00 a.m.\ futures transactions. Right column: Same day from 8:29:40 to 8:32 a.m.. The red horizontal line denotes the regression-based estimated jump size \eqref{reg1} at the event time with 30-seconds pre-average returns. Blue horizontal dashed lines represent the observed price 30 seconds, 5 and 10 minutes after the announcement, respectively.}
\end{figure}

Figure \ref{F:ts} presents the S\&P 500 futures paths on two event days to illustrate the price dynamics around inflation news releases. Panel (A) corresponds to a Regular-news event on May 15, 2024, while Panel (B) corresponds to a Breaking-news event on February 12, 2025. At the Breaking-news event, a dynamic circuit breaker is triggered 1.08 seconds after the news release, causing the price transition to stop for 5 seconds. For each event, the left column displays transaction log-prices over the full 3-hour window centered at 8:30 a.m.\,. The right column zooms into a narrow window from 20 seconds before the release to 2 minutes after. In each zoomed panel, several reference levels are displayed. The blue dashed horizontal lines represent the observed log-price at selected post-announcement times (30 seconds, 5 and 10 minutes after 8:30 a.m.). The red horizontal line represents the regression-based estimate of the efficient jump $\hat bF$. The distance between the red and dashed blue lines indicates the magnitude of the test statistic for fundamental pricing.

In both cases shown in Figure \ref{F:ts}, the observed price moves rapidly at 8:30 a.m.\ and appears to stabilize after 30 seconds. While at Regular-news events prices always appear continuous, prices at Breaking-news events generally exhibit a clear discontinuity. Importantly, the price has only adjusted partially toward the efficient benchmark when trading is interrupted. When trading resumes 5 seconds later, the price continues to adjust toward a new level. This example illustrates that market interventions can happen before the pricing of the fundamental news is fully completed. Using a transition window of 30 seconds, $\H_0$ is not rejected on the Regular-news event on May 15, 2024 shown in Panel (A), but is rejected on the Breaking-news day February 12, 2025.

\subsection{Regressions for estimating jump sizes}\label{subsec:emp2}

The implementation of the regression-based test for fundamental pricing in Proposition \ref{p2} assumes that the jump in the efficient price $J_\tau$ is a linear function of observable factors. Implementing \eqref{reg1}, we estimate cross-event regressions of the $\delta$-pre-averaged event return on the factors that capture the information content of the inflation release. Table \ref{tab:reg} reports the regression results using the $N=55$ Regular-news events. This assumes that fundamental pricing prevails on days with regular news. Since the true termination time of the transition is unknown, we present $\delta$-pre-average event returns with estimation windows of 20 seconds to 10 minutes. 

The regression coefficients shown in Table \ref{tab:reg} appear quite stable across different window sizes considered. Coefficients in bold denote significance at the 1\% level. In all regressions, the Bloomberg surprise variable is significantly negative, indicating that if inflation is higher than markets have expected, the price of the S\&P 500 futures tends to decline. While the Google-attention index is not significant, its interaction with the surprise variable is. The interaction term also has a negative sign, suggesting a more pronounced effect of surprise on event returns when attention is high. 

\begin{table}
\centering
\begin{threeparttable}
\caption{Regressions for jump size estimation\label{tab:reg}}
\begin{tabular}{@{\hskip 0.5cm}l@{\hskip 0.75cm} C{1.4cm}C{1.4cm}C{1.4cm} C{1.4cm}C{1.4cm}C{1.4cm}C{1.4cm}@{\hskip 0.5cm}}
\toprule \hline
\multicolumn{1}{c}{$\delta$ (seconds)} &  20 & 30 & 40 & 60 & 120 & 300 & 600\T\B \\ \hline
Const. &  $\underset{(0.05)}{0.04}$ & $\underset{(0.05)}{0.04}$&$\underset{(0.05)}{0.05}$&$\underset{(0.05)}{0.04}$ &$\underset{(0.05)}{0.04}$ &$\underset{(0.05)}{0.02}$&$\underset{(0.06)}{0.03}$ \T \\
Surp.  &  $\underset{(0.05)}{\bf -0.21}$&$\underset{(0.07)}{\bf -0.24}$&$\underset{(0.07)}{\bf -0.27}$ &$\underset{(0.07)}{\bf -0.29}$ &$\underset{(0.07)}{\bf -0.29}$&$\underset{(0.06)}{\bf -0.27}$&$\underset{(0.08)}{\bf -0.30}$\\
Attn.  &$\underset{(0.05)}{-0.01}$&$\underset{(0.05)}{-0.01}$&$\underset{(0.05)}{-0.02}$ &$\underset{(0.05)}{-0.02}$& $\underset{(0.05)}{-0.03}$&$\underset{(0.06)}{-0.03}$&$\underset{(0.06)}{-0.02}$ \\
Surp.$\times$Attn. &$\underset{(0.07)}{\bf -0.19}$&$\underset{(0.07)}{\bf -0.20}$&$\underset{(0.07)}{\bf -0.21} $&$\underset{(0.07)}{\bf -0.21}$& $\underset{(0.05)}{\bf -0.19}$&$\underset{(0.07)}{\bf -0.22}$&$\underset{(0.06)}{\bf -0.18}$\B \\ \hline
$R^2$  &  0.24&0.34&0.40&0.42& 0.42&0.38&0.37\T \\  \hline
\bottomrule
\end{tabular}
\begin{tablenotes}[flushleft]
\footnotesize
\item Notes: Regressions use $N=55$ Regular-news events. The dependent variable is the $\delta$-pre-averaged event return in percent. The Bloomberg surprise variable (Surp.) and Google attention index (Attn.) are standardized. Heteroskedasticity robust standard errors are in parentheses.
\end{tablenotes}
\end{threeparttable}
\end{table}

The regression fit provides some insights on the transition time. The $R^2$ increases from 24\% for 20-seconds pre-average windows to above 40\% for 2-minute windows. It then declines to 37\% for 10-minute returns. The decline is expected since the residual of the regression is affected by the term $\text{BIAS}_\delta$ in Lemma \ref{l1}, which grows in the integrated squared volatility over $\delta$. A fixed 20-seconds window for all events appears too short to capture completion of transitions. This is particularly true since the longest stop time on Breaking-news events is 20 seconds; see Table \ref{tab:breaks}. It is also important to recognize that the factor structure imposed in Assumption \ref{A3} does not imply that $\delta$-pre-average-event returns are fully explained by observable factors such that we would expect $R^2=1$. The percentage that is not explained through the regression \eqref{reg1} is ascribed to the additional components as shown in Proposition \ref{p2}. That is, the increment of the Brownian martingale component of the efficient price and the two estimation errors coming from the pre-average-return estimation and the regression coefficients, respectively. The median 2.65-fold increase in volatility at the news-release time reported in Table \ref{T:tradestat} suggests that the Brownian component may play a particularly strong role in lowering the $R^2$.

\subsection{Non-fundamental pricing of inflation releases}

We now implement the test for fundamental pricing from Proposition \ref{p2}. The central empirical question is whether the large price movements, particularly at Breaking-news events, are fully justified by fundamentals. Using the regression coefficients from Table \ref{tab:reg}, we estimate the jump in the efficient price for every news releast and compare it against the $\delta$-pre-averaged event return. We vary the event-return window between 20 seconds and 10 minutes. For the critical value, we set the test level to $\alpha=0.01$ and compute the volatility bound $C_{\delta}$ using the same procedure as in the simulation study.

Table \ref{T:main} summarizes the results of the fundamental-pricing test separately for all Regular-news and Breaking-news events. The first two rows report the percentage of events where we reject the null hypothesis of fundamental pricing. For event-return windows between 20 and 60 seconds, we reject fundamental pricing for 29\%--50\% of the Breaking-news events. In other words, 4--7 out of the 14 Breaking-news events display pricing errors that cannot be justified by Brownian fluctuations of the efficient price over the transition interval. News releases that do not trigger a dynamic circuit breaker event can also display non-fundamental pricing. The rejection frequency for Regular-news events is between 7\% and 9\%, which corresponds to  4--5 out of 55 events. Another key result from Table \ref{T:main} is that the detected mispricings at Breaking-news events are accompanied by price overshooting. That is, when prices move up (down), they settle at levels significantly larger (smaller) than the efficient benchmark. This indicates that regulatory interventions into fast-moving markets using dynamic circuit breakers do not provide a safeguard against deviations of asset prices from fundamentals. Instead, triggering events are more likely to generate an extra layer of news that push prices further away from the efficient benchmark.

\begin{table}[t]
\centering
\begin{threeparttable}
\caption{Non-fundamental pricing of inflation releases \label{T:main}}
\label{tab:breaks}
\begin{tabular}{@{\hskip 0.5cm}l@{\hskip 0.5cm}r@{\hskip 0.25cm} R{1.2cm}R{1.2cm}R{1.2cm} R{1.2cm}R{1.2cm}R{1.2cm}R{1.2cm}@{\hskip 0.5cm}}
\toprule \hline
\multicolumn{2}{c}{$\delta$ (seconds)} &  20 & 30 & 40 & 60 & 120 & 300 & 600\T\B \\ \hline
\multirow{2}{*}{Rejection of $\H_0$} & BN & 50\% & 50\% & 36\% & 29\% & 14\% & 7\% & 0\%\T \\
                                     & RN & 9\% & 9\% & 7\% & 4\% & 0\% & 0\% & 0\%\B \\ \cline{2-9}
\multirow{2}{*}{Overshooting}        & BN & 86\% & 100\% & 100\% & 100\% & 100\% & 100\% & --\T \\
                                     & RN & 80\% & 80\% & 75\% & 50\% & -- & -- & --\B \\ \hline
\bottomrule
\end{tabular}
\begin{tablenotes}[flushleft]
\footnotesize
\item Notes:  `RN' denotes Regular-news events and `BN' denotes Breaking-news events. The test from Proposition \ref{p2} is implemented at the significance level $\alpha=0.01$. Overshooting represents the percentage of significant events where regression error and event return have the same sign. 
\end{tablenotes}
\end{threeparttable}
\end{table}

When we use longer windows of at least one minute for the event-return estimation, results in Table \ref{T:main} show that the rejection frequencies fall rapidly to zero. This can be attributed to the possibility that prices revert in the direction toward their fundamental benchmark, such that non-fundamentalness is only a local phenomenon. On the other hand, it can also reflect a natural statistical feature of the test discussed in Section \ref{sec:test}. Longer transition windows inflate the volatility bound in the critical value for testing fundamentalness, and hence make the test more conservative. The impact of $\delta$ on the power of the test has also been illustrated in Figure \ref{Fpower} in the simulation study. In fact, as the critical value scales with $\sqrt \delta$, tests based on a 10-minute transition window use critical value that is 4.5 times larger than for a 30-second window, everything else being equal.       

The scatter plots in Figure \ref{scatter} provide further insights on the test results for fundamental pricing. The horizontal axis reports the regression-implied jump estimate, and the vertical axis reports the $\delta$-pre-averaged event return. The left panel uses an window of 30 seconds, the right panel uses 10 minutes. The black 45-degree line represents perfect fundamental pricing where the regression-implied  jump equals to the $\delta$-pre-average event return. Colors distinguish between Regular-news and Breaking-news events, as well as between events for which fundamentalness is rejected or not rejected. The size of each circle is proportional to the magnitude of the critical value used in the test. Evidently, most Regular-news events cluster tightly around the 45-degree line.

\begin{figure}[t]
\begin{center}
\caption{Jump size at inflation release against $\delta$-pre-average event return.\label{scatter}}
\subfloat[ 30-seconds event returns]
{\includegraphics[width=0.45\textwidth]{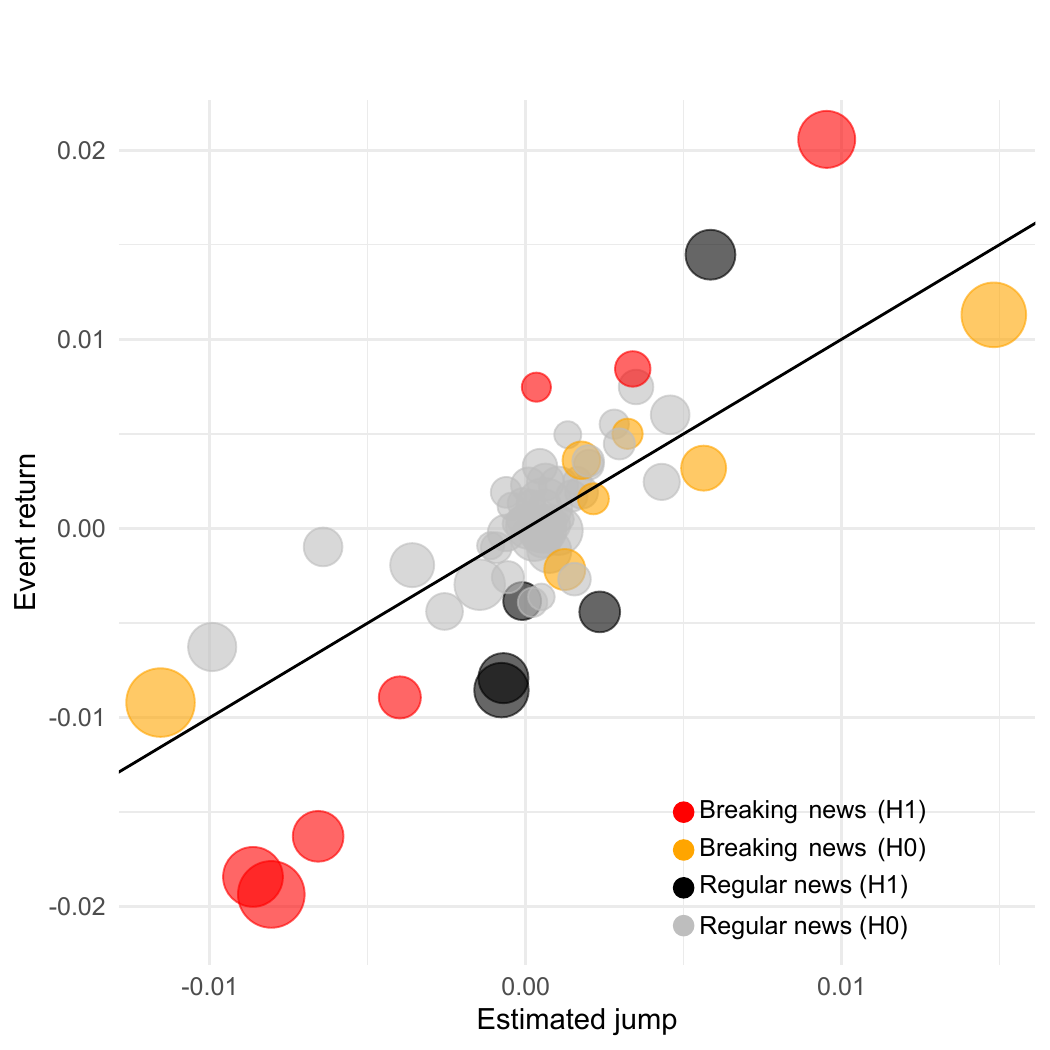}}
\hspace{1cm} 
\subfloat[10-minutes event returns]
{\includegraphics[width=0.45\textwidth]{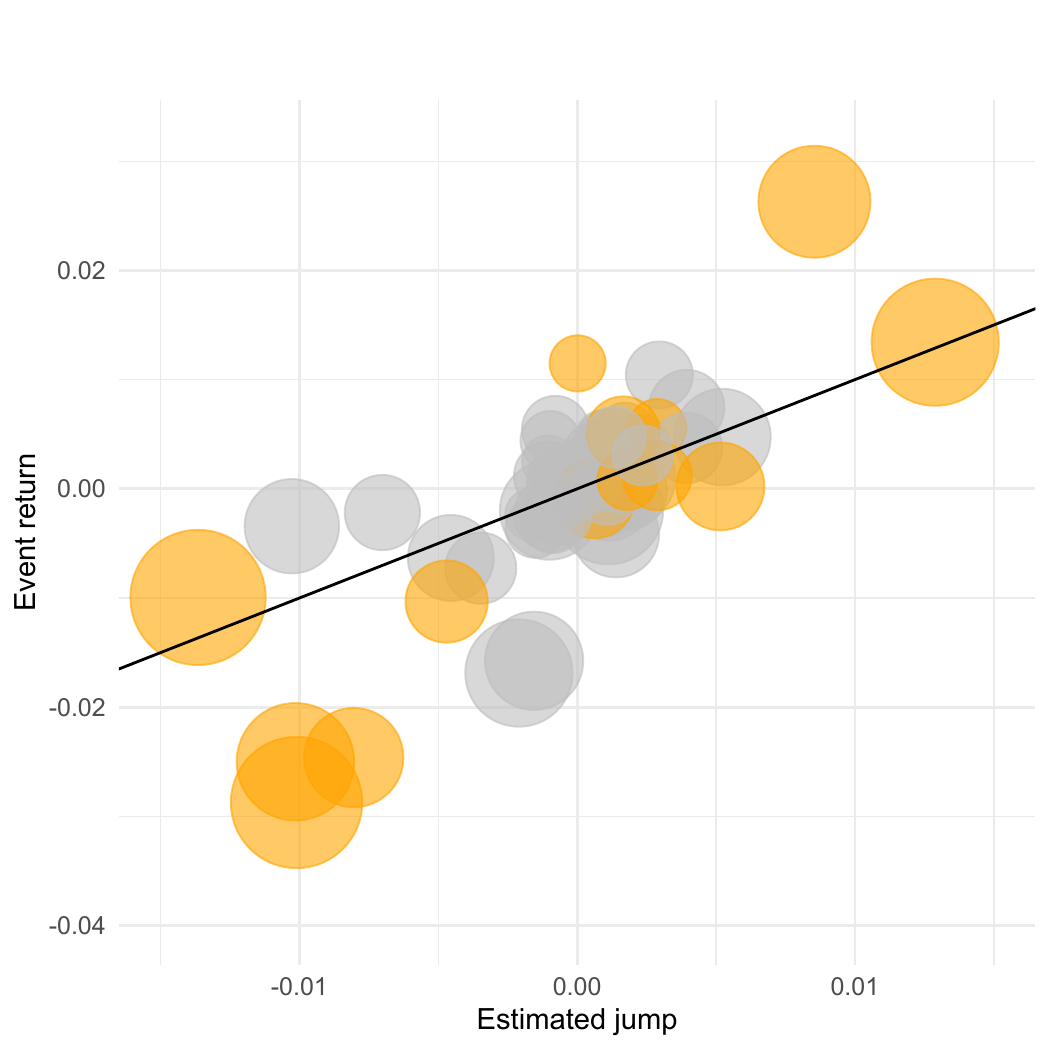}}

\footnotesize
\parbox{17 cm}{\emph{Note:} Each circle refers to one of the 69 inflation news releases between 2020 to 2025.  Applies the test from Proposition \ref{p2}. The diagonal (black line) indicates no pricing and estimation error.  The size of circle is proportional to the magnitude of critical value to test fundamental pricing ($\H_0$) at $\alpha=0.01$. Critical values are about 4 times larger for the 10-minutes window compared with the 30-seconds window. }
\end{center}
\end{figure} 

The left panel highlights the seven Breaking-news events in red where non-fundamental pricing is detected. These events display the largest dispersion from the 45-degree line. In contrast, the right panel based on the 10-minute return window displays no rejections of fundamentalness. Importantly, the disappearance of rejections for larger $\delta$ is not driven by prices reverting toward the fundamental benchmark. If anything, for events with large 30-second returns, the corresponding 10-minute returns tend to deviate even further from the diagonal. Rather, the key change is that the critical values increase sharply with $\delta$, as reflected by the much larger circle sizes in the right panel. Taken together, this visual evidence supports the conclusion in Table \ref{T:main} that longer transition windows eliminate rejections because the post-event volatility surge inflates the critical value more than it increases the initial mispricing.

\section{Conclusion\label{sec:conclude}}

This paper studies a crucial question of whether dynamic circuit breakers distort the pricing of fundamentals. While circuit breaker rules are designed to stabilize markets and prevent disorderly price movements, they mechanically restrict price adjustments when new information arrives. The key issue is therefore whether the ultimate price level reached after the interruption remains consistent with economic fundamentals.

To address this question, we develop high-frequency statistical methods to test for fundamental pricing at news release times. In our model, the efficient price responds instantaneously to news, but the observed price may adjust gradually over a transition interval due to exchange-imposed limits. This non-vanishing transition interval alters the statistical properties of conventional estimators that quantify the impact of the news on efficient prices. 
We show that standard pre-averaging methods for jump estimation become inconsistent, because Brownian fluctuations accumulated over a fixed transition interval generate a non-vanishing bias. The proposed test for fundamental pricing remains valid in the presence of the fixed transition interval and microstructure noise. We implement the test by recovering the jump in the efficient price through a cross-event factor regression.

We apply the framework to high-frequency E-mini S\&P 500 futures data around U.S.\ inflation releases. We find that the triggering of the dynamic circuit breaker shortly after the news release does not guard against non-fundamental pricing. The trading break generally occurs before the price has reached the level that is justified by the content of the news release. Moreover, prices predominantly overshoot where the post-transition price lies beyond the efficient benchmark in the direction of the news. Our results support the interpretation that price distortion is more likely to occur at circuit-breaker-triggering events. Our findings highlight a regulatory trade-off between simple and transparent circuit-breaker rules versus distortions in price discovery. From a practical standpoint, our investigations suggest that more flexible circuit breaker rules are desirable. Ideally, the triggering thresholds can vary with the information flow in the market, allowing for wider bands when substantial news reaches the market. However, in such cases trading breaks lose their current signaling functionality warning about an upcoming discrete move in market prices.   

Our tests are based on fixed transition windows for pre-average return estimation. Since transition times across multiple events may display varying length, it is likely that our fixed windows are conservative. As demonstrated in the simulation study, conservative choices of transition windows can reduce the power of the test. Future research may study the detection of the termination point, which shares similarities with bubble detection in financial data.

\newpage 
{
\singlespacing
\bibliographystyle{apalike2}
\bibliography{CJA}
}

\newpage 
\appendix

\setcounter{lemma}{0}
\renewcommand{\thelemma}{\Alph{section}.\arabic{lemma}}
\renewcommand{\thetheo}{\Alph{section}.\arabic{theo}}

\section{Theoretical results}\label{appd}

 The main purpose is to formally justify the tail bounds used in constructing the test for fundamental pricing and to show how the regression-based correction affects the critical value. The proofs state some fairly standard results which are applied to our setting.

\bigskip

\subsection{\textbf{Proof of Proposition \ref{p1}}}\label{appd:p1}

With reference to the model \eqref{obs} and its underlying assumptions, one can obtain the $\alpha$ critical value $\kappa_{\alpha, \delta}$ generally in three steps.

\bol[(a)]
\li Standard Gaussian tail bound: The upper tail probability for $Z\sim \mathcal{N} (0,v)$  (Mill's-type inequality):
\begin{align*}
\P(|Z|>t) &= \frac{2}{\sqrt{2\pi v}}\int_t^\infty\exp\left(-\frac{z^2}{2v} \right)dz \\
&= \sqrt{\frac{2}{\pi v}}\frac{1}{t}\int_t^\infty t\exp\left(-\frac{z^2}{2v} \right)dz \\
&\le \sqrt{\frac{2}{\pi v}}\frac{1}{t}\int_t^\infty z \exp\left(-\frac{z^2}{2v} \right)dz \\
&= \sqrt{\frac{2}{\pi v}}\frac{v}{t}\int_{t^2/(2v)}^\infty  \exp\left(-u \right)du 
\end{align*}
taking $u=z^2/(2v).$ Integrating gives
\[
\P(|Z|>t)\le \sqrt{\frac{2v}{\pi}}\frac{\exp\left({-t^2}/{2v}\right)}{t}.
\]

\li Use standard Gaussian tail bound together with the bound on integrated squared volatility in Assumption \ref{A1}(i) on the increment in Brownian martingale part over $[\tau,\bar\tau)$:
\[ 
\mathbb P\!\left(|\Delta_{\tau,\delta} X^{(c)}|>\xi\right) \le  \sqrt{\frac{2C_{\delta}}{\pi}}\frac{1}{\xi} \exp\left(-\frac{\xi^2}{2 C_{\delta}}\right). 
\]
The randomness of $\int_\tau^{\bar\tau-}\sigma^2_sds$ is controlled by $C_\delta$.

\li Determine $\xi$ by equating $\sqrt{{2C_{\delta}}/{\pi}}/{\xi} \exp\left(-\frac{\xi^2}{2 C_{\delta}}\right)$ with the test level $\alpha$. Since $\kappa_{\alpha,\delta}$ consists of two independent components (the Brownian martingale component and the $\delta$-pre-average estimation error) we use union bound with level $a=\alpha/2$. We use the Lambert form  $x e^x=y$,   with $x=\xi^2/C_{\delta}$ and $y=2/(\pi a^2)$ and Lambert W-function: $x=W(y)$. Here $x>0$ and $y>0$, so for Lambert we have exactly one real root as the unique real solution: $\xi=\sqrt{C_\delta W(\frac{2}{\pi a^2})}$. 

To eliminate $W$, we let $a\rightarrow 0$, so $y\rightarrow\infty$. Because $y=W(y)\exp(W(y))$ taking log gives 
\[
\log(W(y))+W(y)=\log(y). 
\]
The r.h.s.\ diverges so also the l.h.s. must diverge, implying $W(y)\rightarrow\infty$. Dividing by
$\log(y)$ and studying the terms, the only possibility is $W(y)/log(y)\rightarrow 1$. So we can use $W(y)\approx \log(y)$ (1st order) and rearranging the above equation gives $W(y)\approx \log(y)-\log\log(y)$ (2nd order approx) and 
\[\xi\approx \sqrt{C_{\delta} \left(\log\left(\frac{2}{\pi a^2}\right)-\log\log\left(\frac{2}{\pi a^2}\right) \right)}. \]
Plugging-in $a=\alpha/2$ provides the log terms in Proposition \ref{p1}.
\eol

Since according to Lemma \ref{l1} the $\delta$-pre-average estimation error is normal with random variance $V_{\tau,\bar\tau}$, we follow the same steps but use a conditioning argument in (b) and otherwise replace $C_\delta$ by $V_{\tau,\bar\tau}$ to obtain the final additive bound. 

\qed

\bigskip

\subsection{\textbf{Proof of Proposition \ref{p2}}}\label{appd:p2}

We now derive the critical value for the test with estimated jump size $F\hat b$. That is, to the Brownian and $\delta$-pre-average terms of Proposition \ref{p1}, we add a third component that controls the estimation error of 
\[
\hat b=(F^{*\top}F^*)^{-1}F^{*\top}\Delta^n_{\lfloor \tau n \rfloor,\delta}\hat{X}^*
\] 
in the multiple regression \eqref{reg1}.

Expanding the test statistic, 
\[
\Big| \Delta^n_{\lfloor \tau n \rfloor,\delta}\hat{X} - F\hat b \,\Big| \le \Big|\Delta^n_{\lfloor \tau n \rfloor,\delta}\hat{X} - bF \,\Big|+ \big| F \big| \cdot \big|\hat{b}-b \big|,
\] 
indicates that the first summand on the r.h.s.\ is what we studied in Proposition \ref{p1} but now the test statistic consists of three independent components. To uniformly bound all three, the new joint level is $\alpha/3$. Since $\big|\hat{b} - b\, \big|= \Big|({F}^{*\top} {F}^*)^{-1} {F}^{*\top} {e^*} \Big|$, and $F^*$ is fixed, we only need to control the sum of random elements of $e^*$. Since the error terms are independent, have mean zero and Proposition \ref{p2} provides the bound $c_{e^*}$, we can directly use Hoeffding's inequality for sub-Gaussian random variables. The following expression applies the bound element-wise to the $K$ coefficients in $b$:
\[
\P\!\left( \big|\hat{b} - b \big| > \xi  \right) \le 2 \exp\!\left( -\dfrac{ \xi^2}{2 c_{e^*}^2 \text{diag}({F}^{*\top} {F}^*)^{-1}} \right).
\]
Since the factors $F^*$ are all normalized, $(N\diag({F}^{*\top} {F}^*/N))=N$. Solving for $\xi$ as a function of the desired probability level $\alpha/3$ yields 
\[
\xi = N^{-1/2}\,c_{e^*}\, \sqrt{2 \log\!\left(\frac{6}{\alpha}\right)}\mathbf{1}_K.
\]
All $K$ entries have the same bound such that we can also use $l_1$ norm of $F$ insteat of multiplying by the vector of ones $\mathbf{1}_K$. 

\qed

\setcounter{figure}{0}
\setcounter{table}{0}

\newpage 
\section{Further statistics about Breaking-news events}

\begin{table}[ht]
\centering

\caption{Break dates at inflation news \label{tab:breaks}}
\begin{tabular}{@{\hskip 0.5cm}C{3cm}@{\hskip 0.5cm}C{3cm}C{3cm}C{4cm}@{\hskip 0.5cm}}
\toprule \hline
 Date &\makecell{Time to first\\ break  (seconds)}   &\makecell{ Total stop \\time (seconds)} & 
 \makecell{ Number of trades \\ between 7--10:00 a.m.}    \\ \hline 
2022-06-10 &1.913&5.015& 139,183 \T \\
2022-07-13 &1.183&9.993& 188,426\\
2022-08-10 &1.081&4.996& 131,837 \\
2022-09-13 &0.213&15.90& 75,727\\
2022-10-13 &1.067&20.01& 274,133\\
2022-11-10 &1.063&15.00& 221,903\\
2023-03-14 &1.544&4.998&66,903 \\ 
2024-06-12&0.344&19.99&93,706 \\
2025-01-15&1.236&4.993&162,852 \\
2025-02-12&1.079&4.990& 134,530 \\
2025-03-12&1.196&4.994&183,320 \\
2025-04-10&1.093&15.01&161,041\\
2025-05-13&1.130&5.000&78,756 \\
2025-06-11&1.295&5.002&102,654\B \\ \hline
\bottomrule
\end{tabular}
\medskip 

\footnotesize
\parbox{15.5cm}{Notes: First two columns refer to the transition period starting 8:30 a.m.\ ET. Last column displays $n$. }

\end{table}

\end{document}